\documentclass[12pt]{iopart}

\usepackage[dvips]{graphicx} 
\expandafter\let\csname equation*\endcsname\relax
\expandafter\let\csname endequation*\endcsname\relax
\usepackage{amsfonts,amscd,amsmath,amsthm}
\usepackage{enumerate}
\usepackage{epsfig}
\usepackage{subfigure}
\usepackage{xcolor}
\usepackage[colorlinks = true]{hyperref}
\usepackage{physics}
\usepackage{booktabs}
\usepackage{epstopdf}
\usepackage{framed}
\usepackage{multirow}
\usepackage{color}
\usepackage{longtable}
\usepackage{comment}
\usepackage[ruled,vlined]{algorithm2e}
\usepackage[most]{tcolorbox}
\usepackage[T1]{fontenc}

\graphicspath{{./figure/}}

\usepackage{tikz}
\usetikzlibrary{tikzmark, calc, fit, positioning}
\usetikzlibrary{shapes,arrows}
\usetikzlibrary{quantikz2}

\newtheorem{theorem}{Theorem}
\newtheorem{lemma}{Lemma}
\newtheorem{corollary}{Corollary}

\newtheorem{observation}{Observation}
\newtheorem{definition}{Definition}

\newtcolorbox[auto counter]{mybox}[2][]{
	enhanced,
	breakable,
	colback=blue!5!white,
	colframe=blue!75!black,
	fonttitle=\bfseries,
	title=Box \thetcbcounter: #2,#1
}

\begin{document}
\title{Interplay among entanglement, measurement incompatibility, and nonlocality}
\author{Yuwei Zhu$^{1,2}$, Xingjian Zhang$^2$, Xiongfeng Ma$^{2,*}$}
\address{$^1$ Yau Mathematical Sciences Center, Tsinghua University, Beijing 100084, P.~R.~China}
\address{$^2$ Center for Quantum Information, Institute for Interdisciplinary Information Sciences, Tsinghua University, Beijing 100084, P.~R.~China}
\address{$^*$ Author to whom any correspondence should be addressed.}
\eads{\mailto{xma@tsinghua.edu.cn}}

\begin{abstract}
Nonlocality, manifested by the violation of Bell inequalities, indicates entanglement within a joint quantum system. A natural question is how much entanglement is required for a given nonlocal behavior. Here, we explore this question by quantifying entanglement using a family of generalized Clauser-Horne-Shimony-Holt-type Bell inequalities. Given a Bell-inequality violation, we derive analytical lower bounds on the entanglement of formation, a measure related to entanglement dilution. The bounds also lead to an analytical estimation of the negativity of entanglement. In addition, we consider one-way distillable entanglement tied to entanglement distillation and derive tight numerical estimates. With the additional assumptions of qubit-qubit systems, we find that the relationship between entanglement and measurement incompatibility is not simply a trade-off under a fixed nonlocal behavior. Furthermore, we apply our results to two realistic scenarios --- non-maximally entangled and Werner states. We show that one can utilize the nonlocal statistics by optimizing the Bell inequality for better entanglement estimation.
\end{abstract}


\maketitle

\section{Introduction}
In the early development of quantum mechanics, Einstein, Podolsky and Rosen noticed that the new physical theory leads to a ``spooky action'' between separate observables that is beyond any possible classical correlation~\cite{einstein1935can}. Later, Bell formalizes such a quantum correlation via an experimentally feasible test that is now named after him~\cite{bell1964einstein}. In one of the simplest settings, the Clauser-Horne-Shimony-Holt (CHSH) Bell test~\cite{clauser1969proposed}, two distant experimentalists, Alice and Bob, each has a measurement device and share a pair of particles. While they may not know their devices and physical system \emph{a priori}, they can each take random measurements and later evaluate the Bell expression as shown in Fig.~\ref{fig:CHSH}, 
\begin{equation}
 S=\sum_{a,b,x,y}ab(-1)^{xy}p(a,b|x,y)=\sum_{x,y}(-1)^{xy}\mathbb{E}(ab|x,y),
 \label{CHSH}
\end{equation}
where $x,y\in\{0,1\}$ represent their random choices of measurement settings and $a,b\in\{+1,-1\}$ denote their measurement outcomes, $p(a,b|x,y)$ denotes the outcome probability conditioned on the inputs, and $\mathbb{E}(ab|x,y)$ is the expected value of the product value, $ab$, conditioned on the tuple of inputs, $(x,y)$.
If Alice and Bob observe a value of $S>2$, then they cannot explain the observed correlation using any physical theory that follows local realism. We call such an observation the violation of a Bell inequality, and the statistics exhibit Bell nonlocality. To demonstrate such a nonlocal behavior, the physical systems must exhibit a non-classical feature, wherein quantum theory, entanglement is such an ingredient~\cite{schrodinger1935discussion}. The CHSH expression has a maximal value of $2\sqrt{2}$, which requires the maximally entangled state in a pair of qubits, $\ket{\Phi^+}=(\ket{00}+\ket{11})/\sqrt{2}$~\cite{cirelson1980quantum}. We also term this state the Bell state.

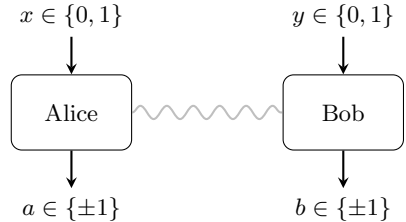
\begin{figure}[hbt!]
\centering
\tikzstyle{arrow} = [thick,->,>=stealth]
\tikzstyle{block} = [rectangle, rounded corners, align=center, minimum width=1.6cm, minimum height=1cm,text centered,draw=black]
\begin{tikzpicture}[node distance = 2cm, auto]	
	\node [block]  (A) {Alice};
	\node [block, right =2cm of A] (B) {Bob};

	\draw[arrow,<-] (A.north) --++ (0,.5) node[above] {$x\in\{0,1\}$};
	\draw[arrow] (A.south) --++ (0,-.5) node[below] {$a\in\{\pm 1\}$};
	\draw[arrow,<-] (B.north) --++ (0,.5) node[above] {$y\in\{0,1\}$};
	\draw[arrow] (B.south) --++ (0,-.5) node[below] {$b\in\{\pm 1\}$};
	\draw[-,decorate,decoration={snake},thick,gray,opacity=0.5] (A.east) -- (B.west);
\end{tikzpicture}
\caption{A diagram of the CHSH Bell test. Two space-like separated users, Alice and Bob, share an unknown quantum state and own untrusted devices. In each round of the CHSH Bell test, Alice applies the measurement determined by the random input $x\in\{0,1\}$ and outputs her measurement result, $a\in\{\pm 1\}$. The measurement process is similar on Bob’s side, with input $y$ and output $b$. As the round of tests accumulated, the CHSH Bell value $S$ in Eq.~\eqref{CHSH} can be evaluated.}
\label{fig:CHSH}
\end{figure}

Entanglement characterizes a joint physical state among multiple parties that cannot be generated through local operations and classical communication (LOCC)~\cite{guhne2009entanglement,horodecki2009quantum}. Beyond its role in understanding quantum foundations, entanglement is a useful resource in a variety of quantum information processing tasks, including quantum communication~\cite{curty2004entanglement}, quantum computation~\cite{jozsa2003role}, and quantum metrology~\cite{giovannetti2011advances}. 
With a resource-theoretic perspective, a large class of information processing operations can be interpreted as entanglement conversion processes under LOCC, wherein one may quantify the participation of entanglement using appropriate measures~\cite{vedral1997quantifying,lo1999unconditional,shor2000simple,bennett1996mixed}.
Therefore, the fundamental question is to detect and quantify entanglement in a system. While state tomography fully reconstructs the information of a quantum state and hence the entanglement properties~\cite{raymer1994complex,leonhardt1996discrete,leonhardt1997measuring}, 
the validity of results relies on the trustworthiness of the detection probes. As detection loss and environmental noise are inevitable in practice, the realistic probes may deviate from the ideal ones~\cite{goh2019experimental}. Even worse, when the malfunction is too severe or the measurements are controlled by adversaries, tomography can lead to false entanglement detection for separable states~\cite{xu2014implementation,yuan2016reliable}.

Fortunately, quantum nonlocality provides a way to bypass the problem. Note that in the Bell test, one does not need to characterize the quantum devices \emph{a priori}, and thus, the indication of entanglement from Bell nonlocality is a device-independent (DI) conclusion~\cite{mayers1998quantum,acin2007device}. This observation leads to the question of what the minimum amount of entanglement is necessary for a given nonlocal behavior.
In other words, Bell tests can serve as a DI entanglement estimation tool. 
In the literature, there are already endeavors into the question~\cite{verstraete2002entanglement,liang2011semi,moroder2013device,toth2015evaluating,arnon2017noise,chen2018exploring,arnon2019device}. The quantitative results provide us with tools for devising novel quantum information processing tasks. A notable investigation is the analysis of DI quantum key distribution~\cite{ekert1992quantum,mayers1998quantum,acin2007device}. With the link among nonlocality, entanglement, and secure communication, we can quantify the key privacy solely from Bell nonlocality~\cite{arnon2018practical,zhang2020efficient,zhang2021quantum}.


Despite the physical intuition for entanglement estimation via Bell nonlocality, the quantitative relation between entanglement and nonlocality can be subtle~\cite{acin2012randomness}. Above all, while the notion of Bell nonlocality arises from the observation of correlations, entanglement is defined by the opposite of a restricted state preparation process. In fact, the Bell nonlocality is a stronger notion than entanglement. Though a nonlocal behavior necessarily requires the presence of entanglement, not all entangled states can unveil a nonlocal correlation~\cite{werner1989quantum,barrett2002nonsequantial}. The conceptual difference even leads to some counter-intuitive results, where a series of works aimed at characterizing their exact relation, such as the discussions on the Peres conjecture --- whether Bell nonlocality is equivalent to distillability of entanglement~\cite{peres1999all,vertesi2014disproving}.

In addition, different entanglement measures may enjoy distinct operational meanings. In general, these measures are not identical to each other. Particularly, there exist quantum states of which the entanglement cost is strictly higher than the distillable entanglement~\cite{lami2023no}. The two measures correspond to the operations of entanglement dilution and entanglement distillation, respectively. Such a phenomenon exhibits the irreversibility of the entanglement theory. Furthermore, estimations of different entanglement measures from the same nonlocal behavior can differ. Taking the CHSH Bell expression as an example, it witnesses a non-trivial value of the negativity of entanglement as long as there is a Bell-inequality violation~\cite{moroder2013device}. On the other hand, the estimation of the negative conditional entropy of entanglement remains zero for a range of the Bell-inequality violation that is not high enough~\cite{arnon2019device}.


In this work, we systematically study entanglement estimation via a family of generalized CHSH-type Bell inequalities. We treat the measurement devices as black boxes. Different implementations can lead to the same observed nonlocal behavior. As depicted in Fig.~\ref{fig:interplay}, a nonlocal behavior necessarily needs both entanglement and incompatible local measurements. In other words, a system with separable states or compatible local measurements definitely fails to observe nonlocality. One may expect a trade-off relationship between state entanglement and measurement incompatibility for a given nonlocal behavior. Hence, we explore the interplay among entanglement, nonlocality, and measurement incompatibility with different entanglement measures. 

\begin{figure}[hbt!]
\centering 
\tikzstyle{arrow} = [thick,->,>=stealth]
\tikzstyle{block} = [rectangle, rounded corners, align=center, minimum width=2cm, minimum height=1cm,text centered,draw=black]
\begin{tikzpicture}[node distance = 2cm, auto]	
	\node[block]  (N) {Nonlocality};
	\node[block] (E) at ($(N)+(220:3cm)$) {Entanglement};
	\node[block] (M) at ($(N)+(-40:3cm)$) {Measurement \\ incompatibility};
	\node[below =.5cm of N] {interplay};

	\draw[arrow,->] (N.south) -- (E.north);
	\draw[arrow,->] (N.south) -- (M.north);
	\draw[latex-latex,thick,red] (E.east) -- (M.west) node[midway,below,black] {\color{red} trade-off};
\end{tikzpicture}
\caption{The interplay among nonlocality, entanglement, and measurement incompatibility. A nonlocal behavior necessarily indicates both entanglement and incompatible local measurements. A system with separable states or compatible local measurements fails in exhibiting nonlocality. Intuitively, under a given nonlocal behavior, one may expect a trade-off relationship between entanglement and measurement incompatibility. In this work, we start from the entanglement estimation via nonlocality, from which we realize that the relation between entanglement and measurement incompatibility is subtler than a simple trade-off. We study the interplay among nonlocality, entanglement, and measurement incompatibility in detail.}

\label{fig:interplay}
\end{figure}
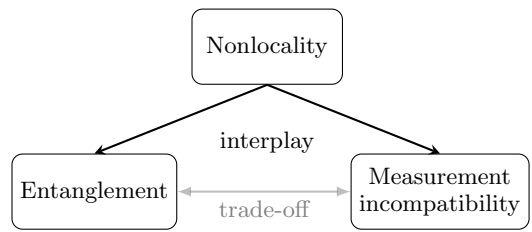

The rest of the paper is organized as follows. In Sec.~\ref{sc:Pre}, we review the necessary concepts in nonlocality and entanglement theories. In Sec.~\ref{sc:framework}, we present the general framework for estimating entanglement in the underlying system using the set of generalized CHSH-type Bell inequalities, namely the tilted CHSH Bell inequalities. Then, we consider three special entanglement measures: the entanglement of formation (EOF), the one-way distillable entanglement, and the negativity of entanglement. For a given Bell-inequality violation, we derive analytical estimation results for the EOF and negativity measures. For one-way distillable entanglement, we obtain tight numerical estimation results. In Sec.~\ref{sc:interplay}, we utilize the entanglement estimation results and investigate the interplay among nonlocality, entanglement, and measurement incompatibility. Particularly, when the underlying state is known to be a pair of qubits, we observe the relation between entanglement and measurement incompatibility under a given nonlocal behavior to be more complex than a simple trade-off. From a practical perspective, in Sec.~\ref{sc:numerical}, we also simulate statistics that arise from pure entangled states and Werner states and examine the performance of our results.

\section{Preliminary and previous work}\label{sc:Pre}
\subsection{General CHSH-type Bell tests}

In this work, we consider the family of generalized CHSH-type Bell tests. Under quantum mechanics, the Bell expression is given by~\cite{acin2012randomness}
\begin{equation}
\label{CHSH_type}
\begin{split}
    S &=\Tr[\rho_{AB}\left(\alpha\hat{A_0}\otimes \hat{B_0}+\alpha \hat{A}_0\otimes \hat{B}_1+\hat{A}_1\otimes \hat{B}_0-\hat{A}_1\otimes \hat{B}_1\right)] \\
    &=\Tr(\rho_{AB}\hat{S}_{\alpha}),
\end{split}
\end{equation}
where $\rho_{AB}$ is the underlying bipartite quantum state, $\hat{A}_{x}$ and $\hat{B}_{y}$ are the observables measured by Alice and Bob, according to their measurement choices, $x,y\in\{0,1\}$, respectively. The family of Bell expressions is parameterized by $\alpha\geq1$, which tilts the contributions of $\hat{A}_0\otimes(\hat{B}_0+\hat{B}_1)$ and $\hat{A}_1\otimes(\hat{B}_0-\hat{B}_1)$ to the Bell value. When $\alpha=1$, Eq.~\eqref{CHSH_type} degenerates to the original CHSH expression defined by Eq.~\eqref{CHSH}~\cite{clauser1969proposed}. For simplicity, we call the expression under the fixed parameter of $\alpha$ as $\alpha$-CHSH expression and $\hat{S}_{\alpha}$ the $\alpha$-CHSH operator. If the underlying quantum state is separable or the local measurement observables are compatible, the $\alpha$-CHSH expression is upper bounded by $S(\alpha)\leq 2\alpha$, which commits a local hidden variable model to reproduce the correlation. In quantum theory, the largest value of $\alpha$-CHSH expression is $2\sqrt{\alpha^2+1}$~\cite{acin2012randomness}. Observation of a Bell value in $(2\alpha,2\sqrt{\alpha^2+1}]$, termed Bell-inequality violation, necessarily implies the existence of entanglement between two local systems and measurement incompatibility between the local measurement observables~\cite{acin2012randomness}. The generalized CHSH-type Bell inequalities have been used for parameter estimation in tasks of DI randomness generation~\cite{acin2012randomness,wooltorton2022tight} and quantum key distribution~\cite{Woodhead2021deviceindependent}, which can outperform the original CHSH inequality in certain practical cases.

In the study of Bell nonlocality, we do not put prior trust in the underlying physical systems. In particular, we do not assume a bounded system dimension. Nevertheless, the simplicity of CHSH-type Bell expressions allows us to apply Jordan's lemma to effectively reduce the system to a mixture of qubit pairs~\cite{acin2007device}. We shall explain how to apply this result when we come to the part of main results. 

\subsection{Entanglement measures}
\label{subsc:ent measures}

In our work, we study entanglement estimation in a bipartite system, $\rho_{AB}\in\mathcal{D}(\mathcal{H}_A\otimes\mathcal{H}_B)$, where we use $\mathcal{D}$ to denote the set of all density operators acting on the associating Hilbert space. 
The first measure we consider is the EOF, $E_{\mathrm{F}}(\rho_{AB})$~\cite{bennett1996mixed}. Operationally, this measure provides a computable bound on the entanglement cost, which quantifies the optimal state conversion rate of diluting maximally entangled states into the desired states of $\rho_{AB}$ under LOCC~\cite{bennett1996mixed}. For a pure state, the EOF equals the entanglement entropy, $E_{\mathrm{F}}(\ket{\phi}_{AB})=H(\rho_A)=H(\rho_B)$, where $\rho_A$ and $\rho_B$ denote the partial state of system $A$ and $B$, respectively, and $H(\cdot)$ represents the von Neumann entropy. When extended to a general state, the EOF is defined via a convex-roof construction,
\begin{equation}
    E_{\mathrm{F}}(\rho_{AB})=\min_{\{p_i,\ket{\phi}_i\}_i}\sum_{i}p_i E_{\mathrm{F}}(\ket{\phi_i}_{AB}),
    \label{EOF}
\end{equation}
where the optimization is taken over all possible pure-state decomposition, $\rho_{AB}=\sum_i p_i \ketbra{\phi_i}_{AB},\sum_ip_i=1,\forall p_i,p_i\geq0$. When restricting $\rho_{AB}$ to the region of two-qubit states, the EOF measure takes a closed form~\cite{hill1997entanglement}, 
\begin{equation}
    E_{\mathrm{F}}(\rho_{AB})=h\left(\frac{1+\sqrt{1-C^2(\rho_{AB})}}{2}\right),
\label{eof}
\end{equation}
where $h(p)=-p\log p-(1-p)\log (1-p)$ is the binary entropy function for $p\in[0,1]$, and $C(\rho_{AB})$ is the concurrence of $\rho_{AB}$, a useful entanglement monotone~\cite{hill1997entanglement,rungta2001universal}. For a general two-qubit quantum state, $\rho_{AB}$, its concurrence is given by
\begin{equation}\label{eq:concurrence}
    C(\rho_{AB})=\max\{0,\lambda_1-\lambda_2-\lambda_3-\lambda_4\},
\end{equation}
where values $\lambda_i$ are the decreasingly ordered square roots of the eigenvalues of the matrix
\begin{equation}
    X(\rho_{AB})=\sqrt{\rho_{AB}}(\sigma_y\otimes\sigma_y)\rho_{AB}^*(\sigma_y\otimes\sigma_y)\sqrt{\rho_{AB}}.
\end{equation}
Here, the density matrix of $\rho_{AB}$ is written on the computational basis of $\{\ket{00},\ket{01},\ket{10},\ket{11}\}$, where $\ket{0}$ and $\ket{1}$ are the eigenstates of $\sigma_z$, and $\rho_{AB}^*$ is the complex conjugate of $\rho_{AB}$.


As opposed to the entanglement dilution process, the entanglement distillation process defines another entanglement measure, the distillable entanglement~\cite{bennett1996mixed}. In this process, given sufficiently many copies of a given state, $\rho_{AB}$, the distillable entanglement is the maximal state conversion rate of distilling maximally entangled states under LOCC. While calculating this measure for a general state remains open, a well-studied lower bound is the one-way distillable entanglement, where classical communication is restricted to a one-way procedure between the two users. In the Shannon limit, where one takes infinitely many independent and identical  copies of the quantum state, the average distillation rate under one-way LOCC can be calculated by the negative conditional entropy~\cite{wilde2017converse} (see Sec.~VIB therein),
\begin{equation}
    E_D^\rightarrow(\rho_{AB})=-H(A|B)_{\rho},
\end{equation}
where $H(A|B)_{\rho}=H(\rho_{AB})-H(\rho_B)$. This result generalizes the finding in Ref.~\cite{bennett1996mixed}, which shows one-way distillable entanglement in the Shannon limit is $1-H(\rho_{AB})$ when $\rho_{AB}$ is a mixture of Bell states (see Sec.~IIIB3 therein). When the underlying state is clear from the context, we shall omit the subscript for simplicity. 

Another entanglement measure we aim to quantify is the negativity of entanglement. This measure is defined in terms of the violation of the positive partial transpose (PPT) criteria~\cite{vidal2002computable},
\begin{equation}
\mathcal{N}(\rho_{AB})=\frac{\|\rho_{AB}^{\mathrm{T}_{A}}\|_1-1}{2}=\sum_{\lambda_i(\rho_{AB}^{\mathrm{T}_{A}})<0}|\lambda_i(\rho_{AB}^{\mathrm{T}_{A}})|,
\label{negativity}
\end{equation}
where $(\cdot)^{\mathrm{T}_{A}}$ is the partial trace operation on subsystem $A$ on the computational basis and $\|\cdot\|_1$ is the trace norm of a matrix. In the second equality of Eq.~\eqref{negativity}, $\lambda_i(\cdot)$ represents the eigenvalues of a matrix. Note that a related measure, namely the logarithm of negativity, $E_{\mathcal{N}}(\rho_{AB})=\log\|\rho_{AB}^{\mathrm{T}_{A}}\|_1$, upper-bounds the distillable entanglement and is hence no less than the negative conditional entropy of entanglement~\cite{guhne2009entanglement}. Notably, the negativity of entanglement is closely related to the concurrence for a pair of qubits given in Eq.~\eqref{eq:concurrence}. Consider the underlying state to be a mixture of two-qubit Bell states,
\begin{equation}
\label{Bell_diagonal}
    \rho_{\lambda}=\lambda_{1}\ketbra{\Phi^+}+\lambda_{2}\ketbra{\Phi^-}+\lambda_{3}\ketbra{\Psi^+}+\lambda_{4}\ketbra{\Psi^-},
\end{equation}
with $\ket{\Phi^\pm}=(\ket{00}\pm\ket{11})/\sqrt{2},\ket{\Psi^\pm}=(\ket{01}\pm\ket{10})/\sqrt{2}$. We term such a state a Bell-diagonal state. Without loss of generality, we assume $\lambda_1\geq\lambda_2\geq\lambda_3\geq\lambda_4$, since we can relabel the eigenvalues corresponding to the Bell-basis states with local unitary operations. The negativity of entanglement for $\rho_{\lambda}$ is given by
\begin{equation}
    \mathcal{N}(\rho_\lambda)=\max\left\{\frac{1}{2}(\lambda_1-\lambda_2-\lambda_3-\lambda_4),0\right\}=\max\left\{\lambda_1-\frac{1}{2},0\right\}.
\end{equation}
Meanwhile, the concurrence of the state is
\begin{equation}
    C(\rho_{\lambda})=\max\{2\lambda_1-1,0\},
\end{equation}
which is exactly twice the negativity of entanglement.



\subsection{Previous work and summary of our results}
In this subsection, we briefly overview the previous findings and summarize our contributions regarding the estimation of entanglement via nonlocality in Table~\ref{tabel:previous_work}. Note that the results in Ref.~\cite{verstraete2002entanglement} and~\cite{liang2011semi} pose an additional assumption on the system dimension, and a part of the results in Ref.~\cite{toth2015evaluating} utilize steering inequalities with full trust on the measurements of one of the parties in a nonlocal setting. Most works simply deal with probabilities, where both the amount of entanglement and the Bell value are taken as expected values. There are a few exceptional works that deal with finite data, including Ref.~\cite{arnon2017noise}, where the authors consider a parallel repetition of a Bell test, and Ref.~\cite{arnon2019device}, where a single-shot estimation of one-way distillable entanglement is given.

In this work, we will estimate the three entanglement measures listed above using the family of tilted CHSH Bell expressions defined in Eq.~\eqref{CHSH_type}. For simplicity, we focus on the expected values. We utilize the tilted CHSH inequalities and obtain tight estimation results. For the negativity of entanglement, we obtain an analytical tight lower bound, proving a conjecture raised from numerical evidence in Ref.~\cite{moroder2013device}. For the EOF, in comparison to the loose estimation in Ref.~\cite{arnon2017noise}, we obtain tight estimation results for the family of tilted CHSH expressions. For the one-way distillable entanglement in the Shannon limit, we obtain tight numerical lower bounds via generalized CHSH Bell expressions in Eq.~\eqref{CHSH_type}. The result for the original CHSH expression coincides with the analytical result obtained in Ref.~\cite{arnon2019device}. 

Note that the DI estimation result is given in terms of the expected values in our study. That is, we present lower bounds on the state entanglement, given the underlying expected Bell value. When implementing entanglement estimation in an experiment, one needs to estimate the Bell value from a finite sample. Also, in a fully DI scenario, the samples may not follow an independent and identical distribution (i.i.d.). For this purpose, one needs to apply statistical methods valid for non-i.i.d. statistics. Notably, the entropy accumulation theorem (EAT) allows us to deal with entropic-based entanglement measures~\cite{arnon2019device}. Using EAT, our one-way distillable entanglement estimation result can be lifted to a finite data-size version over non-i.i.d. statistics~\cite{wilde2017converse} when the Bell test is sequentially repeated in an experiment. In addition, martingale-based techniques may also be applied~\cite{zhang2023quantum}.


\begin{table}[!h]
\begin{center}
\caption{ Entanglement estimation results via nonlocality. In the works that utilize the Navascu\'es-Pironio-Ac\'in-type (NPA-type) hierarchy~\cite{Navascués2008a}, a numerical method, the results numerically converged. The other works give tight bounds on entanglement, except for the results in Ref.~\cite{arnon2017noise}, which utilize the rigidity property of Bell expressions or robust self-testing. }
\label{tabel:previous_work}
\resizebox{\columnwidth}{!}{
\begin{tabular}{ccccc}
\hline
\hline
 Results & Nonlocality Feature & Entanglement measure & Assumption & Main technique\\ 
 \hline
 \cite{verstraete2002entanglement} & CHSH inequality & Concurrence & Dimension & Analytical \\
 \cite{liang2011semi} & Modified CH inequalities & Concurrence & Dimension & Analytical \\ 
 \cite{moroder2013device} & Multipartite Bell inequalities & Negativity & DI& NPA hierarchy\\
\cite{chen2018exploring} & Multipartite Bell inequalities & Robustness of entanglement & DI&NPA hierarchy\\
\cite{toth2015evaluating} & Steering inequalities & Linear entropy & One-sided DI & NPA hierarchy\\
& Bell inequalities &Linear entropy & DI & NPA hierarchy\\
\cite{arnon2017noise} & Threshold quantum games
& EOF & DI & Rigidity\\
\cite{arnon2019device} & CHSH inequality & One-way distillable entanglement & DI & Analytical\\
Our results & Tilted CHSH inequalities & Concurrence & Dimension & Analytical\\
& & EOF & Dimension \& DI & Analytical\\
& & Negativity & Dimension \& DI & Analytical\\
& & One-way distillable entanglement & Dimension \& DI & Numerical\\
\hline
\end{tabular}}
\end{center}
\end{table}




\section{Device-independent entanglement estimation}\label{sc:framework}
\subsection{Entanglement estimation via optimization}
In this section, we formulate the problem of entanglement estimation via Bell nonlocality. Using the nomenclature in quantum cryptography, we also term it DI entanglement estimation. After specifying a particular entanglement measure, $E$, we ask the minimal amount of entanglement in the initial quantum system that supports the observed Bell expression value,
\begin{equation}
\label{original_optm}
\begin{split}
E_{\rm est} &=\min_{\rho_{AB},\hat{A}_0,\hat{A}_1,\hat{B}_0,\hat{B}_1}  E(\rho_{AB}), \\
\text{s.t.}\quad
S &=\Tr\left(\rho_{AB}\hat{S}_{\alpha}\right),\\
\rho_{AB}&\geq 0, \\ \Tr(\rho_{AB})&=1.
\end{split}
\end{equation}
Here, we denote the estimated entanglement measure of $E$ from Bell nonlocality as $E_{\rm est}$. As clarified above, $\hat{S}_{\alpha}$ is the $\alpha$-CHSH operator, an operator function of the measurement observables.

The optimization problem is difficult to solve directly. First, it involves multiple variables, including the underlying quantum state and the measurement observables. Second, the system dimension is unknown, as reflected in Eq.~\eqref{original_optm} where the dimension of $\rho_{AB}$ is unspecified. To address these challenges, we undertake several steps, as illustrated in Fig.~\ref{fig:flowchart}. In the original formulation of Eq.~\eqref{original_optm}, we do not make any assumption on the measurements, which are general measurements characterized by positive operator-valued measures (POVMs). Given that there is no constraint on the system's dimension, Naimark's dilation theorem~\cite{neumark1943representation} allows us to incorporate all local degrees of freedom and extend the measurements to projective ones without any loss of generality. We elaborate this further in \ref{appendix:jordan}.

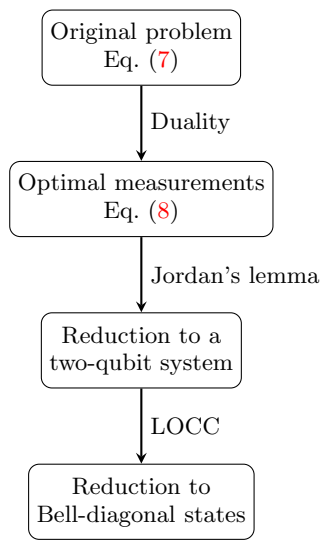
\begin{figure}[hbt!]
    \centering
        \tikzstyle{arrow} = [thick,->,>=stealth]
        \tikzstyle{block} = [rectangle, rounded corners, align=center, minimum width=2cm, minimum height=1cm,text centered,draw=black]
    \begin{tikzpicture}[node distance = 2cm, auto]       
        \node [block]  (b1) {Original problem \\ Eq.~\eqref{original_optm}};
        \node [block, below =1cm of b1] (b2) {Optimal measurements \\ Eq.~\eqref{dual_optm}};
        \node [block, below =1cm of b2] (b3) {Reduction to a\\ two-qubit system};
        \node [block, below =1cm of b3] (b4) {Reduction to \\ Bell-diagonal states};
           \draw [arrow] (b1.south) -- node {Duality} (b2.north);
            \draw [arrow] (b2.south) -- node {Jordan's lemma} (b3.north);
            \draw [arrow] (b3.south) -- node {LOCC} (b4.north);
    \end{tikzpicture}
    \caption{Steps for estimating entanglement via CHSH-type Bell inequalities. Step 1: The original entanglement estimation problem is formulated as Eq.~\eqref{original_optm}. The only constraint is the observed Bell value, $S$. Step 2: Using a duality argument, we consider the optimization problem in Eq.~\eqref{dual_optm}, which can be interpreted as maximizing the Bell value for a given quantum state, $\rho_{AB}$.
    The arguments in the optimal solution are regarded as the ``optimal measurements'' that lead to the maximal Bell value for the state.
    Step 3: By applying Jordan's lemma, we can view the measurement process as resulting from a convex combination of pairs of qubits.
    Step 4: We can further restrict the qubit pairs to Bell-diagonal states in solving the optimization problem. We show that in the CHSH Bell test, any two-qubit state can be transformed to a Bell-diagonal state through LOCC without changing the $\alpha$-CHSH Bell value.} 
    \label{fig:flowchart}
\end{figure}


In the first step, we use duality arguments and transform Eq.~\eqref{original_optm}. Note that the objective function in Eq.~\eqref{original_optm}, $E_{\rm est}:=f(S)$, is continuous and monotonously increasing in its argument $S$, hence having a well-defined inverse function. Consider the following problem,
\begin{equation}
\label{dual_optm}
\begin{split}
S^* &=\max_{\rho_{AB},\hat{A}_0,\hat{A}_1,\hat{B}_0,\hat{B}_1}  \Tr\left(\rho_{AB}\hat{S}_{\alpha}\right), \\
\text{s.t.}\quad
E(\rho_{AB})&=E_{\rm est}, \\
\rho_{AB}&\geq0, \\
\Tr(\rho_{AB})&=1,
\end{split}
\end{equation}
where the objective function in Eq.~\eqref{dual_optm} comes from the inverse function of the original optimization problem, $S^*:=f^{-1}(E_{\rm est})$. In solving Eq.~\eqref{dual_optm}, as the objective function is bilinear in $\rho_{AB}$ and Bell operator $\hat{S}_{\alpha}$, the optimization equals the maximization over the two arguments individually, $S^* =\max_{\rho_{AB}} \max_{\hat{A}_0,\hat{A}_1,\hat{B}_0,\hat{B}_1}  \Tr\left(\rho_{AB}\hat{S}_{\alpha}\right)$. 
For the inner optimization, denote $S^*(\rho_{AB})=\max_{\hat{A}_0,\hat{A}_1,\hat{B}_0,\hat{B}_1} \Tr\left(\rho_{AB}\hat{S}_{\alpha}\right)$, which can be seen as the maximal $\alpha$-CHSH Bell value that can be obtained with $\rho_{AB}$. Then, we may equivalently solve Eq.~\eqref{original_optm} with the following optimization, 
\begin{equation}
\label{equivalent_optm}
\begin{split}
E_{\rm est}&=\min_{\rho_{AB}}E(\rho_{AB}), \\
\text{s.t.}\quad
S^*(\rho_{AB}) &= S,\\
\rho_{AB}&\geq0, \\
\Tr(\rho_{AB})&=1.
\end{split}
\end{equation}
For simplicity, we call the measurements that yield the maximal $\alpha$-CHSH Bell value for $\rho_{AB}$ the ``optimal measurements''.
\begin{definition}
The optimal measurements of state $\rho_{AB}$ are the observables that maximize the $\alpha$-CHSH expression in Eq.~\eqref{CHSH_type} for $\rho_{AB}$, i.e., $\operatorname{argmax}_{\hat{A}_0,\hat{A}_1,\hat{B}_0,\hat{B}_1} \Tr\left(\rho_{AB}\hat{S}_{\alpha}\right)$.
\end{definition}

To bypass the dimension problem, we utilize Jordan's lemma. 
We leave the detailed analysis in \ref{appendix:jordan}. Here, we briefly state the indication of Jordan's lemma in our work. In the CHSH-type Bell test, we can effectively view the measurement process as first performing local operations to transform the underlying quantum state into an ensemble of qubit pairs, $\{p^\mu,\rho_{AB}^{\mu}\}$, with $p^{\mu}$ a probability distribution, and then measuring each pair of qubits with associate qubit observables. The measurement on each pair of qubits corresponds to a Bell value, $S^{\mu}$, and the observed Bell value is the average of these values, $S=\sum_{\mu}p^{\mu}S^{\mu}$. Guaranteed by the convexity property of an entanglement measure, we can lower-bound the amount of entanglement in the initial system by studying the average amount of entanglement in the ensemble of qubit-pairs, $\sum_{\mu}p^{\mu}E(\rho_{AB}^{\mu})$. In this way, we can essentially focus on quantifying entanglement in a pair of qubits. 

By further utilizing the non-increasing property under LOCC of an entanglement measure and choosing proper local computational bases, we may further restrict the pair of qubits to a Bell-diagonal state in Eq.~\eqref{Bell_diagonal} for simplicity. We have the following lemma.
\begin{lemma}
    In a CHSH Bell test, under a fixed computational basis, a two-qubit state, $\rho_{AB}$, can be transformed into a Bell-diagonal state, $\rho_\lambda$, through LOCC, with the $\alpha$-CHSH Bell value unchanged.
\label{lemma:LOCC}
\end{lemma}

The lemma indicates that an observed Bell value can always be interpreted as arising from a Bell-diagonal state. Furthermore, the operations in the lemma are restricted to LOCC and state mixing. Since these operations do not increase entanglement, we can hence restrict our analysis of lower-bounding entanglement to the set of Bell-diagonal states. The LOCC transformation in this result was first constructed in Ref. \cite{Pironio2009device} (see Lemma 3 therein). Here, we verify the unchanged $\alpha$-CHSH Bell value through the LOCC transformation. We present proof of the lemma in \ref{appendix:Bell-diagonal}. 

With the above simplifications, we have the following lemma in solving the problem in Eq.~\eqref{equivalent_optm} and leave the proof in \ref{appendix:Bell-diagonal}.

\begin{lemma}
The maximal value of the $\alpha$-CHSH expression in Eq.~\eqref{CHSH_type} for a Bell-diagonal state shown in Eq.~\eqref{Bell_diagonal}, $\rho_{\lambda}$, is given by
\begin{equation}
    S= 2\sqrt{\alpha^2(\lambda_{1}+\lambda_{2}-\lambda_{3}-\lambda_{4})^2+(\lambda_{1}-\lambda_{2}+\lambda_{3}-\lambda_{4})^2},
\end{equation}
where $\lambda_i$ is the $i$-th largest eigenvalue of $\rho_{\lambda}$.
\label{lemma:maximal_value}
\end{lemma}

In our analysis, we extend general measurements to projective ones by Naimark's dilation theorem. For two projection-valued measurements (PVMs), their incompatibility is defined as the largest inner product of their eigenvectors. Suppose two projective measurements are given by observables $\hat{M}$ and $\hat{N}$, and $\{\ket{m}\}$ and $\{\ket{n}\}$ are their eigenvectors, respectively. Then we define the incompatibility as
\begin{equation}
\label{eq:inc_overlap}
    I(\hat{M},\hat{N})=\max_{\ket{m},\ket{n}} |\braket{m}{n}|.
\end{equation}
For qubit observables, this definition can be equivalently given by the observable commutator.
Consider two qubit observables $\hat{M}=\cos\theta_1\sigma_z+\sin\theta_1\sigma_x$ and $\hat{N}=\cos\theta_2\sigma_z+\sin\theta_2\sigma_x$. Then, Eq.~\eqref{eq:inc_overlap} becomes 
\begin{equation}
\label{eq:inc_PVM}
    I(\hat{M},\hat{N})=\frac{1}{2}\abs{\cos(\theta_1-\theta_2)}+\frac{1}{2},
\end{equation}
On the other hand, the commutator between them is $[\hat{M},\hat{N}]=\sin(\theta_1-\theta_2)[\sigma_x,
    \sigma_z]$. Considering the symmetry, the commutator essentially provides a quantity as
\begin{equation}
\label{eq:com_PVM}
    I_c(\hat{M},\hat{N})=\abs{\sin(\theta_1-\theta_2)}[\sigma_x,
    \sigma_z].
\end{equation}
We can see $I(\hat{M},\hat{N})$ and $I_c(\hat{M},\hat{N})$ are of one-to-one bijection. Therefore, we use the commutator as a direct incompatibility measure in the qubit PVM case. And specifically, we use the coefficient, $\abs{\sin(\theta_1-\theta_2)}$, in Eq.~\eqref{eq:com_PVM} in the later incompatibility discussions. In proving Lemma~\ref{lemma:maximal_value}, a notable issue is that the optimal measurements may not be the most incompatible measurements. Up to a minus sign before the observables, the optimal measurements of the Bell-diagonal state in Eq.~\eqref{Bell_diagonal} are as follows,
\begin{equation}
\label{Bell_opt_meas}
 \begin{split}  
       \hat{A}_0&=\sigma_z,\\
       \hat{A}_1&=\sigma_x,\\
       \hat{B}_0&=\cos\theta\sigma_z+\sin\theta\sigma_x,\\
       \hat{B}_1&=\cos\theta\sigma_z-\sin\theta\sigma_x,
 \end{split}
\end{equation}
where $\theta$ fully determines the amount of imcompatibility of the local observables, with $\tan\theta=(\lambda_{1}-\lambda_{2}+\lambda_{3}-\lambda_{4})/[\alpha(\lambda_{1}+\lambda_{2}-\lambda_{3}-\lambda_{4})]$ determined by the Bell-diagonal state and parameter $\alpha$. While the observables on Alice's side are maximally incompatible with each other, the commutator of the observables on Bob's side is given by $[\hat{B}_0,\hat{B}_1]=\sin2\theta[\sigma_x,\sigma_z]$.
For example, when the considered state is the maximally entangled state with $\lambda_1=1,\lambda_2=\lambda_3=\lambda_4=0$ and $\alpha=1$, which corresponds to the original CHSH expression, the optimal measurements coincide with the most incompatible measurements. For more cases where $\sin2\theta$ is strictly smaller than $1$, $\hat{B}_0$ and $\hat{B}_1$ are not maximally incompatible.

\begin{observation}
The observables that yield the largest $\alpha$-CHSH Bell value for a quantum state are not the most incompatible ones in general.
\end{observation}

Notwithstanding, a subtle issue is that we do not have access to the underlying probability distribution in the qubit-pair ensemble, $p^{\mu}$, or the underlying Bell value for each pair of qubits. As we only know the average Bell value over the ensemble, we need to be careful of convexity issues. Suppose the solution to Eq.~\eqref{original_optm} with the restriction of a pair of qubits takes the form $E_{\rm est}=E_{\rm est}(S)$. 
When extending the result to possibly an ensemble of qubit pairs, if $E_{\rm est}(S)$ is not concave in $S$, then 
\begin{equation}
\label{convexity}
    E_{\rm est}\left(\sum_{\mu}p^{\mu}S^{\mu}\right)\leq \sum_{\mu}p^{\mu}E_{\rm est}\left(S^{\mu}\right)\leq\sum_\mu p^{\mu} E\left(\rho_{AB}^\mu\right)\leq E(\rho_{AB}),
\end{equation}
which holds for any probability distribution $p^{\mu}$. Hence, we can directly lower-bound the amount of entanglement in the underlying state by $E_{\rm est}(S)$, where $S=\sum_{\mu}p^{\mu}S^{\mu}$ represents the observed Bell value. Yet if the function $E_{\rm est}(S)$ is concave, namely 
$[E_{\rm est}(S_1)+E_{\rm est}(S_2)]/2<E_{\rm est}[(S_1+S_2)/2]$, then the first inequality in Eq.~\eqref{convexity} no longer holds valid. Consequently, we need to take a ``convex closure'' of the function $E_{\rm est}$ to estimate the amount of entanglement from a quantum state with an unknown dimension. Here, we explain the concept of convex closure in our context.
Suppose a concave function $f(x)$ is defined on the interval $[a,b]$. Then the convex closure of $f(x)$, denoted as $f_{\mathrm{con}}(x)$, is given by
    \begin{equation}
    \label{eq:convex_closure}
        f_{\mathrm{con}}(x)=\frac{f(b)-f(a)}{b-a}(x-a)+b,
    \end{equation}
    which represents a straight line connecting points $(a,f(a))$ and $(b,f(b))$.

Another implicit issue is that we assume the entanglement measure to have a consistent definition for all dimensions, such that the last inequality in Eq.~\eqref{convexity} holds. Yet for the measure of concurrence, its definition in a high-dimensional system is subtle. 
Despite this, we may estimate the average amount of concurrence of the qubit pairs arising from the block-dephasing operation in the measurement, $\sum_{\mu}p^{\mu}C_{\rm est}\left(S^{\mu}\right)$.

Following the above discussions, we study the entanglement measures of EOF and one-way distillable entanglement, which are essentially given by concurrence and conditional entropy of entanglement, respectively. 

\subsection{Concurrence and entanglement of formation}
In this subsection, we take concurrence $C(\cdot)$ as the objective entanglement measure in Eq.~\eqref{original_optm}. For this measure, we have an analytical estimation result.
\begin{theorem}\label{thm:concurrence_est}
    Suppose the underlying quantum state is a pair of qubits. For a given tilted CHSH expression in Eq.~\eqref{CHSH_type} parametrized by $\alpha$, if the Bell expression value is $S$, then the amount of concurrence in the underlying state can be lower-bounded,
\begin{equation}
\label{concurrence_est}
    C(\rho_{AB})\geq \sqrt{\frac{S^2}{4}-\alpha^2}.
\end{equation}
The equality can be saturated when measuring a Bell-diagonal state in Eq.~\eqref{Bell_diagonal} with eigenvalues
\begin{equation}
    \begin{gathered}
        \lambda_1=\frac{1}{2}+\frac{1}{2}\sqrt{\frac{S^2}{4}-\alpha^2},\\
        \lambda_2=\frac{1}{2}-\frac{1}{2}\sqrt{\frac{S^2}{4}-\alpha^2},\\
        \lambda_3=\lambda_4=0,
    \end{gathered}
\end{equation}
using measurements in Eq.~\eqref{Bell_opt_meas} with $\theta=\arctan(\frac{1}{\alpha}\sqrt{\frac{S^2}{4}-\alpha^2})$.
\end{theorem}

We leave the detailed derivation in  \ref{appendix:concurrence_est}. 

\begin{observation}
Given a $\alpha$-CHSH Bell value, the measurements that require the minimum entanglement are not the most incompatible measurements in general.
\end{observation}

As the EOF can be expressed by concurrence in a closed form for a pair of qubits~\cite{hill1997entanglement}, this entanglement measure is directly lower-bounded by substituting Eq.~\eqref{concurrence_est} in Eq.~\eqref{eof},
\begin{equation}
\label{EOF_est_semiDI}
    E_{\rm F}(\rho_{AB})\geq h\left(\frac{1}{2}+\frac{1}{2}\sqrt{1+\alpha^2-\frac{S^2}{4}}\right).
\end{equation}

\begin{figure}[hbt!]
\centering
\includegraphics[scale=0.57]{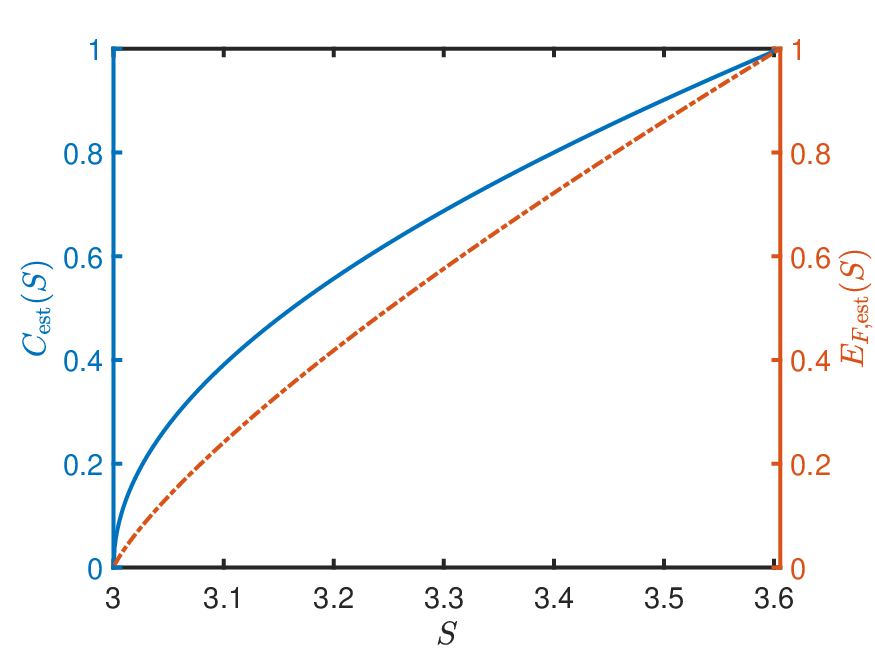}
\caption{Diagram of concurrence and entanglement formation estimation results when the CHSH-type expression in Eq.~\eqref{CHSH_type} takes $\alpha=1.5$ and input states are two-qubit states. We plot the estimated values of concurrence and EOF
with the blue solid line and the red dashed line, respectively. The estimations are both concave in $S\in(3,2\sqrt{3.25}]$ and range from $0$ to $1$.}
\label{fig:concurrence_EOF}
\end{figure}

In Fig.~\ref{fig:concurrence_EOF}, we depict the entanglement estimation result when $\alpha=1.5$ for a pair of qubits input. Given the consistent definition of EOF across all dimensions, we extend the two-qubit EOF estimation result in Eq.~\eqref{EOF_est_semiDI} to a general-state scenario. Since the two-qubit EOF estimation result is a concave function in the Bell value, a convex closure should be taken when extending the EOF estimation result to general states with an unknown dimension. For example, suppose the underlying state already has a block-diagonal form with respect to the measurement observables, $p^{(1)}\rho_{AB}^{(1)}\oplus p^{(2)}\rho_{AB}^{(2)}$, where $p^{(1)}=p^{(2)}=1/2$, and $\rho_{AB}^{(1)},\rho_{AB}^{(2)}$ are qubit pairs. In addition, $\rho_{AB}^{(1)}$ is not entangled, while $\rho_{AB}^{(2)}$ is a Bell state. In this case, the EOF of the underlying state is $E_{\mathrm{F}}(\rho_{AB})=p^{(1)}E_{\mathrm{F}}(\rho_{AB}^{(1)})+p^{(2)}E_{\mathrm{F}}(\rho_{AB}^{(2)})=1/2$. Suppose the measurement observables are such that the expected Bell values arising from $\rho_{AB}^{(1)}$ and $\rho_{AB}^{(2)}$ are $S^{(1)}=3$ and $S^{(2)}=2\sqrt{1+1.5^2}$, respectively. Alice and Bob can only observe the average expected Bell value of $S=p^{(1)}S^{(1)}+p^{(2)}S^{(2)}$, and they will overestimate the underlying state's EOF if they directly apply Eq.~\eqref{EOF_est_semiDI}. To bypass such a problem, we take a convex closure over Eq.~\eqref{EOF_est_semiDI} according to Eq.~\eqref{eq:convex_closure} and obtain the final estimation.


\begin{theorem}
For a given tilted CHSH expression in Eq.~\eqref{CHSH_type}, if the Bell expression value is $S$, then the amount of entanglement of formation in the underlying state can be lower-bounded,
\begin{equation}
    E_{\rm F}(\rho_{AB})\geq\frac{S-2\alpha}{2\sqrt{1+\alpha^2}-2\alpha}.
    \label{EOF_est_DI}
\end{equation}
\end{theorem}

In the literature, Ref.~\cite{arnon2017noise} provides EOF estimation results using threshold games that have a non-zero gap between classical and quantum strategies. Translating the result to the CHSH game, the EOF estimation result in Ref.~\cite{arnon2017noise} is
\begin{equation}
\label{EOF_arnon}
    E_{\mathrm{F}}(\rho_{AB})\geq\frac{(S-2)^5}{10\cdot 180^2\cdot2^{16}}.
\end{equation}
In comparison, our EOF estimation result in Eq.~\eqref{EOF_est_DI} is much tighter. In  \ref{sc:EOF_arnon}, we briefly review the results in Ref.~\cite{arnon2017noise} and explain how to arrive at Eq.~\eqref{EOF_arnon}.

\subsection{Negative conditional entropy and one-way distillable entanglement}
In this subsection, we estimate the one-way distillable entanglement, $E_D^\rightarrow(\rho_{AB})$, depicted by the negative conditional entropy, $-H(A|B)$, via Bell nonlocality. 
For the set of Bell-diagonal states on the qubit-pair systems, since the reduced density matrix of a subsystem is a maximally mixed state, $H(B)=1$, the conditional von Neumann entropy of the state is reduced to $H(A|B)=H(AB)-H(B)=H(AB)-1$. Using the notation in Eq.~\eqref{Bell_diagonal}, the 
term of joint von Neumann entropy can be expressed by
\begin{equation}
    H(AB)=H(\Vec{\lambda})=-\sum_{i=1}^{4}\lambda_i\log\lambda_i.
\end{equation}
Thus, the lower bound of one-way distillable entanglement for a pair of qubits becomes the following optimization problem,
\begin{equation}
    \begin{split}
    E_{D\rm ,est}^\rightarrow & =\underset{\lambda_i,i=1,2,3,4}{\text{min}} 1+\sum_{i=1}^4\lambda_i\log\lambda_i, \\
    \text{s.t.}\quad
    S &=2\sqrt{\alpha^2(\lambda_1+\lambda_2-\lambda_3-\lambda_4)^2+(\lambda_1-\lambda_2+\lambda_3-\lambda_4)^2},\\
    \lambda_1 & \geq\lambda_2\geq\lambda_3\geq\lambda_4,\\
    1 & =\sum_{i=1}^4\lambda_i, \lambda_i\geq 0 \; ,i=1,2,3,4.
    \end{split}
    \label{entropy_optm}
\end{equation}
As this is a convex optimization problem, we can solve it efficiently via off-the-shelf numerical toolboxes. We present numerical results for some values of $\alpha$ in Fig.~\ref{fig:numeric_entropy}. Given any $\alpha>1$, the estimation value $E_{D\rm ,est}^\rightarrow(S)$ is a convex function on $S\in(2\alpha,2\sqrt{1+\alpha^2}]$. Following Eq.~\eqref{convexity}, the solution can be directly lifted as the lower bound on one-way distillable entanglement for a general state. 
Notably, in the special case of $\alpha=1$ that corresponds to the original CHSH expression, our numerical estimation coincides with the existing analytical result~\cite{arnon2019device},
\begin{equation}
    E_{D}^{\rightarrow}(\rho_{AB})=-H(A|B)\geq \max\left\{0,1-2 h\left(\frac{1}{2}-\frac{S}{4 \sqrt{2}}\right)\right\}.
\end{equation}

\begin{figure}[hbt!]
    \centering
    \includegraphics[scale=0.5]{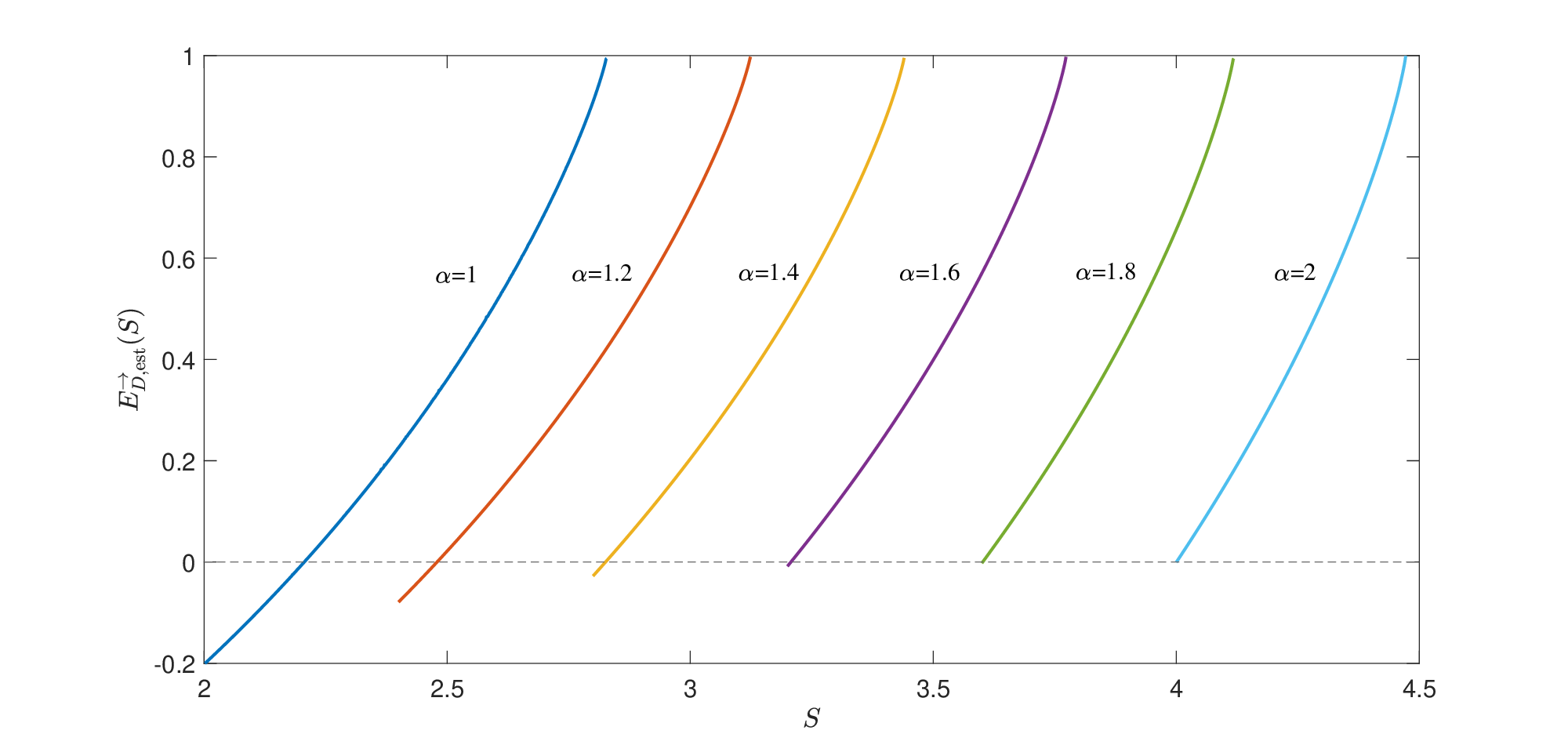}
    \caption{Diagram of one-way distillable entanglement estimation results. The estimation is depicted by CHSH-type Bell expressions with several discretely increasing $\alpha$. For each value of $\alpha$, the estimation result, $E_{D\rm ,est}^\rightarrow(S)$, depicted over the valid interval $S\in(2\alpha,2\sqrt{1+\alpha^2}]$ is convex. When $\alpha$ increases, $E_{D\rm ,est}^\rightarrow(S)$ at $S=2\alpha$ for each $\alpha$ increases and converges to $0$. Since $E_{D,\rm est}^\rightarrow(S)$ is convex in $S$, the estimation results hold valid without assuming the system dimension.}
    \label{fig:numeric_entropy}
\end{figure}


\subsection{Negativity of entanglement}
In this subsection, we estimate the negativity of entanglement given an $\alpha$-CHSH Bell value. In Sec.~\ref{subsc:ent measures}, we establish a connection between the concurrence and negativity for Bell-diagonal states: both measures are linear functions of eigenvalues $\lambda_i$ in Eq.~\eqref{Bell_diagonal}, and the negativity of entanglement is precisely half the concurrence for a given Bell-diagonal state. Therefore, for a pair of qubits, we can obtain the following analytical estimation result for negativity based on the concurrence estimation in Theorem~\ref{thm:concurrence_est}.

\begin{corollary}
    Suppose the underlying state is a pair of qubits. For a given tilted CHSH expression in Eq.~\eqref{CHSH_type}, if the Bell expression value is $S$, then the amount of negativity in the underlying state can be lower-bounded,
\begin{equation}
\label{negativity_est}
    \mathcal{N}(\rho_{AB})\geq \frac{1}{2}\sqrt{\frac{S^2}{4}-\alpha^2}.
\end{equation}
The equality can be saturated when measuring a Bell-diagonal state in Eq.~\eqref{Bell_diagonal} with eigenvalues
\begin{equation}
    \begin{gathered}
        \lambda_1=\frac{1}{2}+\frac{1}{2}\sqrt{\frac{S^2}{4}-\alpha^2},\\
        \lambda_2=\frac{1}{2}-\frac{1}{2}\sqrt{\frac{S^2}{4}-\alpha^2},\\
        \lambda_3=\lambda_4=0,
    \end{gathered}
\end{equation}
using measurements in Eq.~\eqref{Bell_opt_meas} with $\theta=\arctan(\frac{1}{\alpha}\sqrt{\frac{S^2}{4}-\alpha^2})$. 
\end{corollary}

In a fully DI scenario, similar to the estimation of concurrence, we take a convex closure to lower-bound the negativity of the underlying state with an unknown dimension, arriving at the following analytical result.

\begin{corollary}
    For a given tilted CHSH expression in Eq.~\eqref{CHSH_type}, if the Bell expression value is $S$, then the amount of negativity in the underlying state can be lower-bounded,
    \begin{equation}
        \label{negativity_est_DI}
            \mathcal{N}(\rho_{AB})\geq \frac{S-2\alpha}{4(\sqrt{1+\alpha^2}-\alpha)}.
        \end{equation}
\end{corollary}

Especially when $\alpha=1$, Eq.~\eqref{negativity_est_DI} becomes
\begin{equation}
\label{negativity_est_DI_alpha1}
    \mathcal{N}(\rho_{AB})\geq \frac{S-2}{4(\sqrt{2}-1)}.
\end{equation}
This analytical result proves the conjecture of Eq.~(5) in Ref.~\cite{moroder2013device}, where the authors observed a nearly linear relation between the lower bound of negativity and the underlying Bell value via the third level of an NPA-type hierarchy numerical algorithm~\cite{navascues2007bounding}. In Fig.~\ref{fig:semiDI_negativity}, we plot the negativity estimation results from the CHSH Bell value, i.e., $\alpha=1$, in the fully DI scenario in Eq.~\eqref{negativity_est_DI} and the semi-device-independent (semi-DI) case with the additional assumption of system dimensions in Eq.~\eqref{negativity_est}.

\begin{figure}[hbt!]
\centering
\includegraphics[scale=0.57]{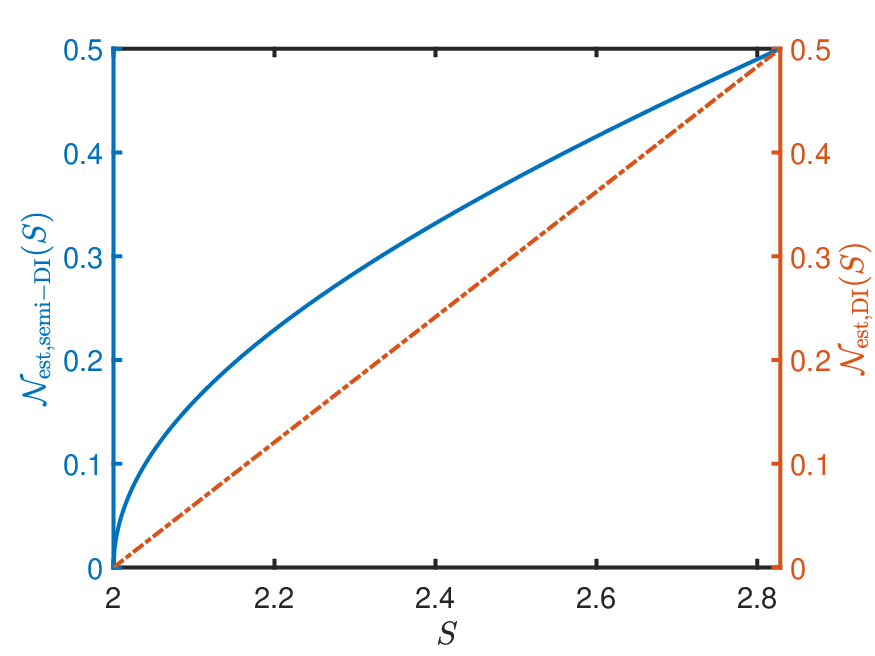}
\caption{Diagram of negativity estimation results. The semi-DI negativity estimation when the input state is a pair of qubits is plotted with the blue solid line. The complete DI negativity estimation result is plotted with the red dashed line.}
\label{fig:semiDI_negativity}
\end{figure}

\section{Exploration of the relation among entanglement, measurement incompatibility, and Bell nonlocality within a finite-dimensional system}\label{sc:interplay}

Besides entanglement, another key ingredient behind nonlocality is measurement incompatibility. Both entanglement and measurement incompatibility can be regarded as quantum resources to unveil non-classical physical phenomena; hence a natural intuition is that for a given Bell value, there is a trade-off relation between entanglement and measurement incompatibility, where more incompatible measurement may compensate an underlying system with less entanglement and \emph{vice versa}. However, as we have discussed for the notion of optimal measurements, the observables that yield the largest Bell value for a quantum state may not correspond to the maximally incompatible ones. Particularly, as shown in Theorem~\ref{thm:concurrence_est}, for the case of the least amount of entanglement for a nonlocal behavior, the observables are generally not maximally incompatible.
In this section, we make a detailed investigation into the relation between entanglement and measurement incompatibility under a given Bell nonlocal behavior. 


To simplify the discussion and manifest the main factors in the interplay, we restrict our analysis with the following assumptions: (1) the underlying system is a pair of qubits, and (2) the measurement operators are qubit observables. Note this setting is consistent with the DI estimation. Namely, this is the subsystem after applying Naimark’s dilation theorem and Jordan’s lemma. Moreover, when considering projective measurements, the quantification of measurement incompatibility becomes straightforward.


We quantify the least amount of entanglement that is necessary for a given Bell value,
\begin{equation}
\label{interplay_optm}
    \begin{split}
     E_{\mathrm{est}}(S,\theta)&=\min_{\rho_{AB}}  E(\rho_{AB}), \\
    \text{s.t.}\quad
    \Tr\left(\rho_{AB}\hat{S}_{\alpha}\right) &= S,\\
   \hat{A}_0 &= \sigma_z,\\
   \hat{A}_1 &= \sigma_x,\\
    \hat{B}_0 &=\cos\theta\sigma_z+\sin\theta\sigma_x,\\
    \hat{B}_1 &=\cos\theta\sigma_z-\sin\theta\sigma_x,\\
    \rho_{AB}&\geq 0, \\
    \rho_{AB}&\in\mathcal{D}(\mathcal{H}_2\otimes\mathcal{H}_2), \\
    \Tr(\rho_{AB})&=1.
    \end{split}
\end{equation}
where $E$ represents a chosen entanglement measure. We still denote the solution to the optimization as $E_{\mathrm{est}}$, while it now represents the least amount of entanglement that is necessary for the nonlocal behavior under the given measurement incompatibility. In this optimization, we assume that the measurement incompatibility is parameterized by one parameter, $\theta$. On Alice's side, the two local observables are fixed to be maximally incompatible with each other. On Bob's side, when $\theta=0$ and $\pi/2$, the two local observables commute. When $\theta=\pi/4$, the local observables enjoy the maximal incompatibility, which is the other extreme. In the following discussions, we restrict the parameter to be $\theta\in[0,\pi/4]$, as other cases can be obtained via symmetry.

Before presenting the results, we make some remarks on the scenario considered in the optimization. For both parties in the Bell test, their local measurement observables need to be incompatible to violate a Bell inequality. Aside from our choice of fixing Alice's observables to be maximally incompatible while optimizing Bob's observables, one may consider alternative settings of incompatible measurements and analyze their relation with entanglement and nonlocality. The reason for our choice is that the measurement settings in Eq.~\eqref{interplay_optm} coincide with those in Eq.~\eqref{Bell_opt_meas}, which are the optimal measurements that give the largest Bell value for the Bell-diagonal states in Eq.~\eqref{Bell_diagonal}. Moreover, if we further minimize $E_{\rm est}(S,\theta)$ over $\theta$ in Eq.~\eqref{interplay_optm}, the optimization degenerates to Eq.~\eqref{equivalent_optm} with the additional assumption that the underlying system is a pair of qubits and the measurements are projective.

\subsection{Original CHSH (\texorpdfstring{$\alpha=1$}{Lg})}
To observe the interplay among entanglement, measurement incompatibility, and nonlocality, we numerically solve the optimization problem in Eq.~\eqref{interplay_optm} by taking concurrence and one-way distillable entanglement as an entanglement measure and varying $\theta\in[0,\pi/4]$ and $S$ discretely from the $\alpha$-CHSH Bell values.

In Fig.~\ref{fig:concurrence-incompatibility}, we choose the original CHSH Bell expression and present the numerical results when taking the concurrence as the entanglement measure. For a nonlocal behavior, where $S\in (2,2\sqrt{2}]$, we denote $\theta=\theta_C^*$ when the estimated concurrence reaches its minimum, $C_{\rm est}(S,\theta)=\sqrt{S^2/4-1}$. As we have derived in Theorem~\ref{thm:concurrence_est}, $\theta_C^*=\arctan\sqrt{S^2/4-1}$. 
When the amount of measurement incompatibility between the local observables is smaller than that of this point, which corresponds to $\theta<\theta_C^*$, there is a trade-off relation between concurrence and measurement incompatibility, where less entanglement is required for the given Bell value as the amount of measurement incompatibility increases.
However, when $\theta>\theta_C^*$, as the underlying state enjoys more entanglement of concurrence, larger measurement incompatibility is also required for the observed Bell value. As $S$ increases from $2$ to $2\sqrt{2}$, the range of feasible values of $(\theta,C_{\mathrm{est}})$ shrinks as $S$ grows. When $S=2\sqrt{2}$, the underlying state is maximally entangled and the local measurement observables are the most incompatible ones, and the range of possible values of $(\theta,C_{\mathrm{est}})$ degenerates to the point of $(\pi/4,1)$. This result coincides with the self-testing finding~\cite{bardyn2009device}, where the only feasible experimental setting for the maximum CHSH Bell value enjoys the above properties.


\begin{figure}[hbt!]
    \centering
    \includegraphics[scale=0.57]{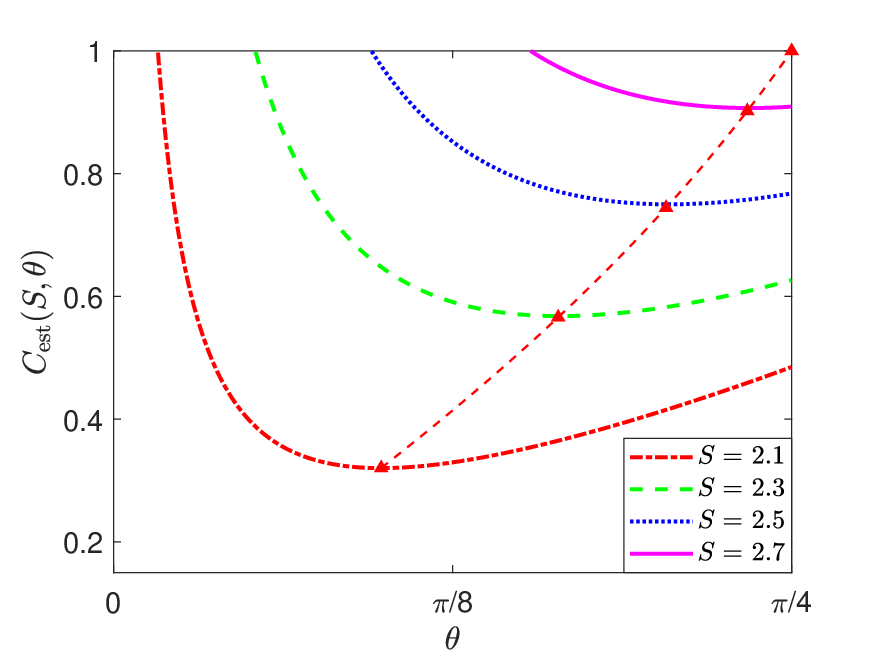}
    \caption{Illustration of the interplay among Bell nonlocality, measurement incompatibility, and concurrence. In this figure, we consider the original CHSH Bell expression and parameterize the measurement observables as in Eq.~\eqref{Bell_opt_meas}, where incompatibility is quantified through $\theta$. We focus on the interval of $\theta\in[0,\pi/4]$, and the results elsewhere can be obtained using symmetry. 
    For a given value of $S$, when $\theta<\theta_C^*=\arctan\sqrt{S^2/4-1}$, there is a trade-off relation between entanglement and measurement incompatibility, where less entanglement of concurrence is required for the nonlocal behavior when the measurements become more incompatible and \emph{vice versa}. Afterward, more entanglement is required for the given Bell value as $\theta$ increases. As $S$ increases, the range of possible values of $(\theta,C_{\mathrm{est}})$ shrinks and $\theta_C^*$ gets close to $\pi/4$. 
    }
    \label{fig:concurrence-incompatibility}
\end{figure}

In Fig.~\ref{entropy-incompatibility}, we present the numerical results when taking the one-way distillable entanglement as the entanglement measure. Under a fixed Bell value, there is a strict trade-off relation between entanglement and measurement incompatibility. The more incompatible the measurement observables are, the less entanglement is necessary for the nonlocal behavior, and \emph{vice versa}. In addition, the range for the trade-off shrinks with a larger Bell violation value. In the extreme of the largest Bell violation value, $S=2\sqrt{2}$, the setting should involve both the maximally entangled state and measurement observables that are maximally incompatible, in accordance with the self-testing result. One thing to note is that the estimated negative conditional entropy reaches its minimum exactly when $\theta=\pi/4$ for all $S\in(2,2\sqrt{2}]$, which holds no longer valid in $\alpha$-CHSH inequality when $\alpha>1$.

\begin{figure}[hbt!]
    \centering
    \includegraphics[scale=0.57]{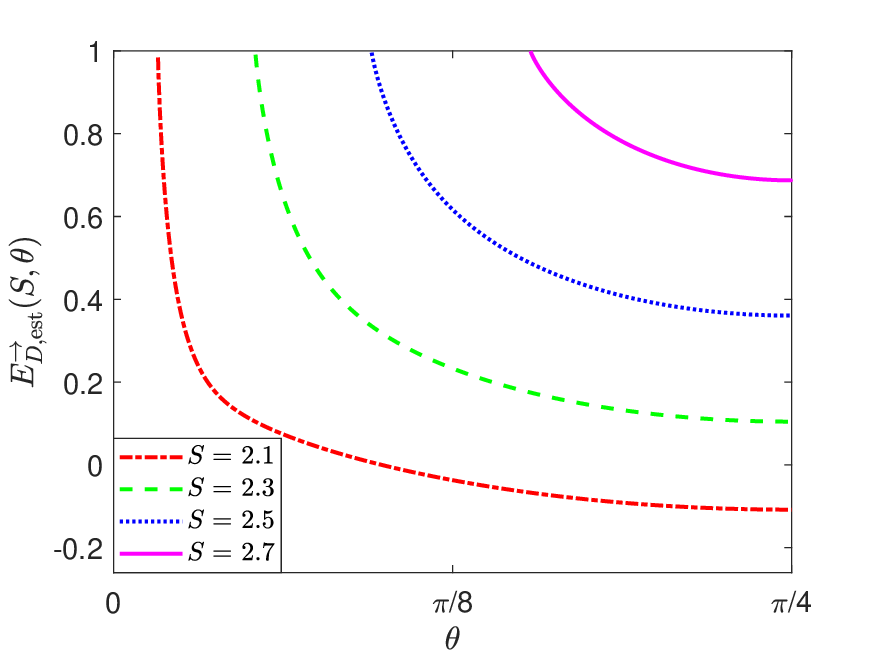}
    \caption{Illustration of the interplay among Bell nonlocality, measurement incompatibility, and one-way distillable entanglement. In this figure, we consider the original CHSH Bell expression and parameterize the measurement observables as in Eq.~\eqref{Bell_opt_meas}.
    As $S$ increases, the range of possible values of $(\theta,E_{D,{\rm est}}^\rightarrow)$ shrinks. For a given Bell value, less entanglement is required when $\theta$ increases in the valid region.}
\label{entropy-incompatibility}
\end{figure}



\subsection{General CHSH-type (\texorpdfstring{$\alpha>1$}{Lg})}

Besides the original CHSH Bell expression, we also study the relation among entanglement, measurement incompatibility, and nonlocality for general $\alpha$-CHSH expressions. Fixing parameter $\alpha>1$, for any Bell value $S\in(2\alpha,2\sqrt{1+\alpha^2}]$, denote the range of plausible values of parameter $\theta$ by $\theta_{\min}\leq\theta\leq\theta_{\max}$. In Fig.~\ref{incompatibility_alpha12}, we investigate the issue under parameter $\alpha=1.2$. For both the concurrence of entanglement and one-way distillable entanglement, when $\theta$ increases from $\theta_{\min}$ to $\theta_{\max}$, the corresponding amount of estimated entanglement first monotonically decreases from $1$, which corresponds to the maximally entangled state. In this region, there is a trade-off relation between entanglement and measurement incompatibility under the given Bell value. After reaching its minimum at $\theta=\theta^*_{E}$, a point that is related to the particular entanglement measure under study, more entanglement is required as the local measurement observables become more incompatible.
One thing worth noting is that under the same $S$, the values of $\theta_{\min}$ and $\theta_{\max}$ are the same for both entanglement measures we now study. As $S$ grows, the supported range of incompatibility and entanglement shrinks, which converges to the single point of $\theta=\arctan(1/\alpha)$ and $E_{\rm est}=1$ when $S$ approaches its maximum $2\sqrt{1+\alpha^2}$. Namely, the maximum value of the $\alpha$-CHSH expression requires a pair of non-maximally incompatible measurements on one side. This result also coincides with the self-testing findings \cite{acin2012randomness}. Another indication is that to yield a large $\alpha$-CHSH Bell value with $\alpha>1$, the measurement observables on one side cannot be too incompatible, where they lie outside the feasible region of the experimental settings.

\begin{figure}[hbt!]
    \centering
    \includegraphics[scale=0.42]{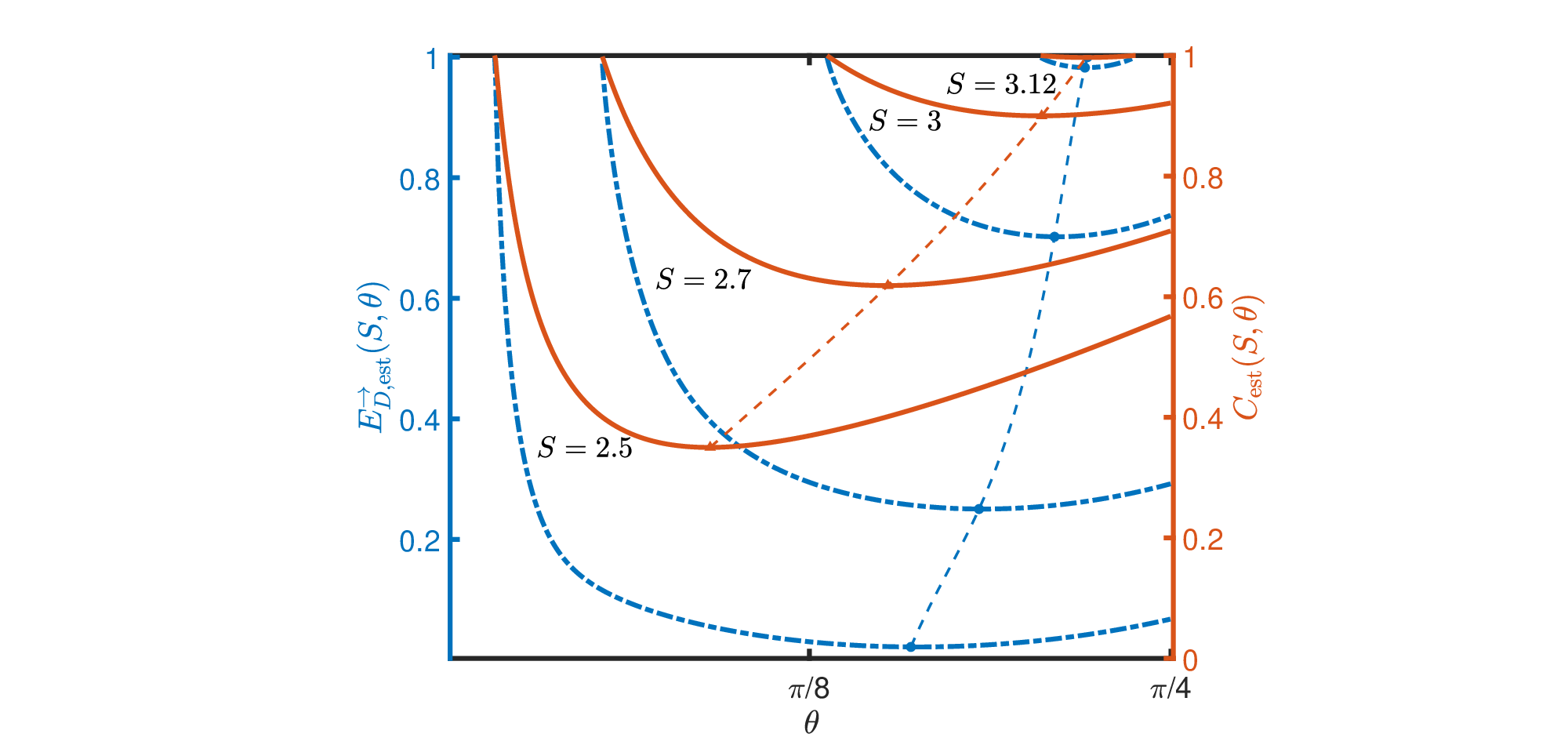}
    \caption{Illustration of the interplay among Bell nonlocality, measurement incompatibility, and entanglement. In this figure, we consider the $\alpha$-CHSH Bell expression with  $\alpha=1.2$. The blue curves depict the results of one-way distillable entanglement, and the red curves depict the results of concurrence. For both entanglement measures, given a Bell value, the least required amount of entanglement first monotonically decreases as $\theta$ increases. After $\theta$ is larger than a threshold value that depends on the entanglement measure, $\theta_E^*$, more entanglement is required as the measurements become more incompatible.
    The ranges of possible values of $\theta\in[\theta_{\min},\theta_{\max}]$ are the same for the two entanglement measures. 
    When $S<2.2\sqrt{2}$, $\theta_{\max}=\pi/4$. When $S\geq 2.2\sqrt{2}$, $\theta_{\max}$ is smaller than $\pi/4$. The supported range shrinks as $S$ increases. When $S$ reaches its maximum, $S=2\sqrt{1.2^2+1}$, the range degenerates to the point of $\theta=\arctan 1/1.2$. In this case, the underlying state can only be a maximally entangled state, corresponding to $E_{\rm est}=1$.}
\label{incompatibility_alpha12}
\end{figure}

For concurrence, we can derive the critical points analytically. Given $\alpha$-CHSH Bell value $S$, when $\theta=\theta^{*}_{C}=\arctan(\frac{1}{\alpha}\sqrt{\frac{S^2}{4}-\alpha^2})$, the system requires the least amount of concurrence, $C_{\mathrm{est}}=\sqrt{\frac{S^2}{4}-\alpha^2}$, which can be derived from Theorem~\ref{thm:concurrence_est}. When $\theta>\theta_{C}^*$, we find there is a region of $\theta$ where the least amount of concurrence behaves differently from that of one-way distillable entanglement. That is, though more concurrence is required in the underlying system as $\theta$ grows, the system may yield less distillable entanglement. In other words, the manifestation of entanglement properties through nonlocality highly depends on the particular entanglement measure under study.

The value of $\theta_{\max}$ and the value of corresponding $E_{\mathrm{est}}$ are related to the parameter, $\alpha$. A notable issue is that under particular value of $\alpha$ and Bell value $S$, $E_{\mathrm{est}}$ at $\theta=\theta_{\max}$ can reach $1$. We find that when $1<\alpha<\sqrt{2}+1$, for $S<\sqrt{2}(\alpha+1)$, $\theta_{\max}=\pi/4$ and the corresponding least amount of entanglement, $E_{\rm est}$ is strictly smaller than $1$. For a larger Bell value, $S\geq\sqrt{2}(\alpha+1)$, $\theta_{\max}$ may be smaller than $\pi/4$, and $E_{\rm est}$ at $\theta=\theta_{\max}$ always reaches $E_{\rm est}= 1$. For Bell expressions with $\alpha\geq \sqrt{2}+1$, as long as the Bell inequality is violated, $S>2\alpha$, we have $E_{\rm est}= 1$ at $\theta=\theta_{\max}$. 
In Fig.~\ref{incompatibility_alpha0}, we illustrate the interplay relation when $\alpha=\sqrt{2}+1$. From this example, we can see that there can be two experimental settings that give rise to the same Bell value, where the underlying systems enjoy the same amount of entanglement, yet the incompatibility between the local measurements can be significantly different.


\begin{figure}[hbt!]
    \centering
    \includegraphics[scale=0.42]{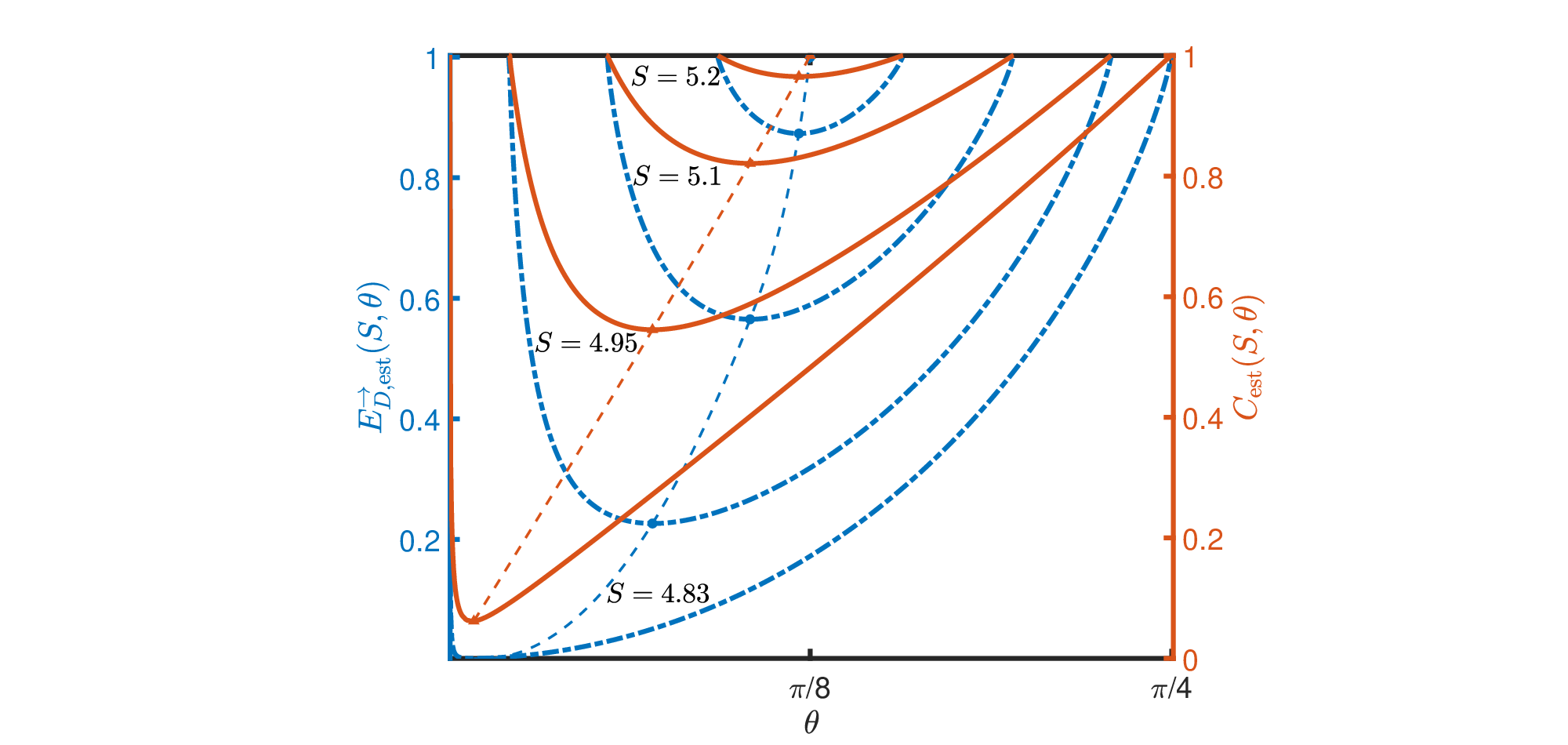}
    \caption{Illustration of the interplay among Bell nonlocality, measurement incompatibility, and entanglement. In this figure, $\alpha=\sqrt{2}+1$. The blue curves depict the results of one-way distillable entanglement, and the red curves depict the results of concurrence. The relation between entanglement and measurement incompatibility is similar to that in Fig.~\ref{incompatibility_alpha12}. Nevertheless, given any Bell value $S$ that is larger than $2\alpha$, which violates the $\alpha$-CHSH Bell inequality, the least amount of entanglement in the system at $\theta=\theta_{\max}$ is $1$, corresponding to the maximally entangled state. The feasible range of $\theta\in[\theta_{\min},\theta_{\max}]$ shrinks as $S$ grows and degenerates to the point of $\theta=\arctan(\sqrt{2}-1)$, where the Bell value reaches its maximum, $S=2\sqrt{2\sqrt{2}+4}$.}
\label{incompatibility_alpha0}
\end{figure}

\section{Optimizing entanglement estimation in realistic settings}\label{sc:numerical}
While the full probability distribution of a nonlocal behavior gives the complete description in a Bell test, for practical purposes, one often applies a Bell expression to characterize nonlocality.
As a given Bell expression only reflects a facet of the nonlocal behavior, one may expect a better entanglement estimation result via some well-chosen Bell expressions. In particular, realistic experiments unavoidably suffer from loss and noise in state transmission and detection. The robustness of such imperfections can differ for various Bell inequalities. In this section, we aim to specify when a non-trivial choice of $\alpha$-CHSH expression leads to better estimation. 
From an experimental point of view, the investigations may benefit experimental designs of DI information processing tasks. 
For this purpose, we simulate the nonlocal correlations that arise from two sets of states: Non-maximally entangled pure states and Werner states. The deliberate use of non-maximally entangled states has been proved beneficial for observing nonlocal correlations under lossy detectors~\cite{eberhard1993background}. The Werner states characterize the typical effect of transmission noise upon entanglement distribution through fiber links~\cite{liu2018device,li2021experimental}.

With respect to the computation bases that define Pauli operators $\sigma_z$ on each local system, the measurements are parametrized as
\begin{equation}
\label{meas_setting}
\begin{split}
    \hat{A}_0&=\sigma_z, \\
    \hat{A}_1&=\cos\theta_1\sigma_z+\sin\theta_1\sigma_x, \\
    \hat{B}_0&=\cos\theta_2\sigma_z+\sin\theta_2\sigma_x, \\
    \hat{B}_1&=\cos\theta_3\sigma_z+\sin\theta_3\sigma_x,
\end{split}
\end{equation}
for Alice and Bob, respectively. We examine the optimal choice of $\alpha$ for DI entanglement estimation if the statistics arise from the two types of states.

\subsection{Non-maximally entangled states}
In the first simulation model, the underlying state is a non-maximally entangled state. We express the state on its Schmidt basis,
\begin{equation}
    \ket{\phi_{AB}(\delta)}=\cos\delta\ket{00}+\sin\delta\ket{11}.
\label{pure}
\end{equation}
where parameter $\delta\in[0,\pi/2]$ fully determines the amount of entanglement in the system. We first present the estimation result through a concrete example. We specify the underlying system by $\delta=\pi/6$ and the measurements by $\theta_1=\pi/2,\theta_2=\pi/6$ and $\theta_3=-\pi/6$. As shown in Fig.~\ref{fig:pure_entropy_concurrence}, we estimate the amount of negative conditional entropy and EOF with respect to the simulated statistics. The estimation results vary with respect to the value of $\alpha$.
The curves show that the original CHSH expression, corresponding to $\alpha=1$, does not yield the best entanglement estimation result for the given statistics. One obtains the best estimation results with the value of $\alpha$ roughly in the range $[1.4,1.6]$ for both EOF and negative conditional entropy.



\begin{figure}[hbt!]
    \centering
    \includegraphics[width=0.48\textwidth]{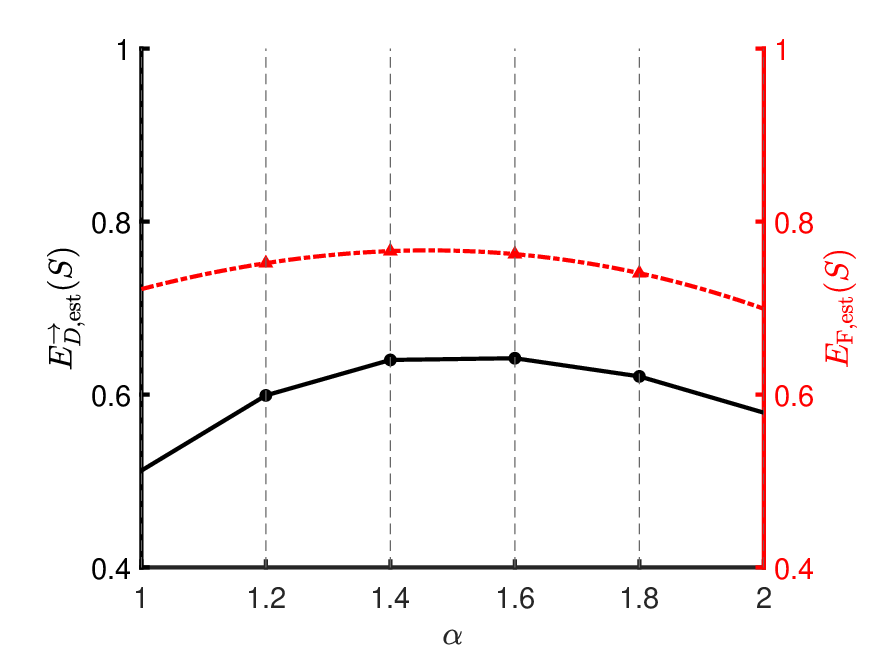}
    \caption{Entanglement estimation results for nonlocal correlations arising from non-maximally entangled states. The experimental setting is given by $\delta=\pi/6,\theta_1=\pi/2,\theta_2=\pi/6$ and $\theta_3=-\pi/6$. We depict the entanglement estimation results when using different $\alpha$-CHSH Bell expressions. We plot the estimated values of one-way distillable entanglement and EOF with the black solid line and the red dashed line, respectively.
    }
\label{fig:pure_entropy_concurrence}
\end{figure}

To see when better entanglement estimation is obtained with $\alpha>1$ for the family of non-maximally entangled states, we analytically derive the condition of the underlying system for the measure of EOF.
Using Eq.~\eqref{EOF_est_DI}, we have the following result.

\begin{theorem}\label{thm:DI_condition_pure}
    In a Bell test experiment, suppose the underlying state of the system takes the form of Eq.~\eqref{pure}, and the observables take the form of Eq.~\eqref{meas_setting}. For EOF estimation solely from the violation values of $\alpha$-CHSH Bell inequalities, if $\theta_1,\theta_2,\theta_3$ and $\delta$ satisfy
\begin{equation}
    \sin2\delta\sin\theta_1(\sin\theta_2-\sin\theta_3)+\cos\theta_2(\sqrt{2}+1+\cos\theta_1)+\cos\theta_3(\sqrt{2}+1-\cos\theta_1)>2(1+\sqrt{2}),
    \label{eq:DI_condition_pure}
\end{equation}
then there exists $\alpha>1$, where a better estimation of $E_{\rm F,est}(S)$ can be obtained by using the $\alpha$-CHSH inequality parameterized by this value than by using the original CHSH inequality (corresponding to $\alpha=1$).
\end{theorem}
Theorem~\ref{thm:DI_condition_pure} analytically confirms the nonlocality depicted by the original CHSH Bell value does not always provide the EOF estimation that approaches the real value most. When a fixed nonlocal behavior is given in a CHSH Bell test, once the non-maximally entangled state parameter $\delta$ in Eq.~\eqref{pure} and measurement parameters $\theta_1,\theta_2,\theta_3$ in Eq.~\eqref{meas_setting} satisfy Eq.~\eqref{eq:DI_condition_pure}, it is feasible to take a CHSH-type Bell value with $\alpha>1$ to estimate the EOF of the state. We leave the proof of Theorem~\ref{thm:DI_condition_pure} in  \ref{appendix:numerical}.

\noindent \textbf{Example.}
We take a special set of parameters in Eq.~\eqref{eq:DI_condition_pure} for an example. Suppose $\theta_1=\pi/2$ and $\theta_3=-\theta_2$, which resemble the optimal measurements in Eq.~\eqref{Bell_opt_meas} in form. Under this setting, we derive an explicit expression of $\alpha_0>1$, such that the estimation $E_{\rm F,est}(S)$ is optimal when taking $\alpha=\alpha_0$ in the CHSH inequality. When $0<\theta_2<\pi/4$, any non-maximally entangled state that satisfies
\begin{equation}
    \sin2\delta>(1+\sqrt{2})\frac{1-\cos\theta_2}{\sin\theta_2}
    \label{DI_condition_pure_delta}
\end{equation}
permits a better EOF estimation characterizing with some $\alpha>1$. When the state and measurements satisfy the condition in Eq.~\eqref{DI_condition_pure_delta}, one obtains the optimally estimated EOF when the parameter $\alpha$ equals
\begin{equation}
    \alpha_E^*=\frac{1}{2}\left(T-\frac{1}{T}\right)>1,
\end{equation}
where we denote $T=\frac{\sin2\delta\sin\theta_2}{1-\cos\theta_2}$. The optimally estimated EOF is then given by
\begin{equation}
    E_{\rm F,est}|_{\alpha_E^*}=\frac{1-\cos\theta_2}{2}(T^2+1).
\end{equation}
It is worth mentioning that if we have the additional assumption that the underlying state is a pair of qubits, we can analytically derive a more accurate estimation result of EOF. We leave the detailed conclusions and examples in  \ref{appendix:numerical}.

\subsection{Werner states}
In the second simulation model, we consider the set of Werner states,
\begin{equation}
\label{werner}
    \rho_{\mathrm{W}}(p)=(1-p)\ketbra{\Phi^+}+p\frac{I}{4},
\end{equation}
where we write $\ket{\Phi^+}=(\ket{00}+\ket{11})/\sqrt{2}$. The Werner state is entangled when $p<2/3$. 
Similarly, for the family of Werner states, there are examples that a non-trivial choice of $\alpha$-CHSH expression gives a better estimation result. In Fig.~\ref{fig:werner_entropy_concurrence}, we present such an example. In the simulation, the underlying system is parameterized by $p=0.05$, and the measurements are parameterized by $\theta_1=\pi/2,\theta_2=\pi/6$ and $\theta_3=-\pi/6$. The optimal estimation of negative conditional entropy is obtained with the value of $\alpha$ roughly in the range of $[1.2,1.4]$, while the optimal estimation of EOF is obtained when $\alpha\in [1,1.2]$. We derive an analytical result for the feasible region of state and measurement parameters that permits a better EOF estimation for a non-trivial value of $\alpha>1$.

\begin{figure}[hbt!]
    \centering
    \includegraphics[width=0.48\textwidth]{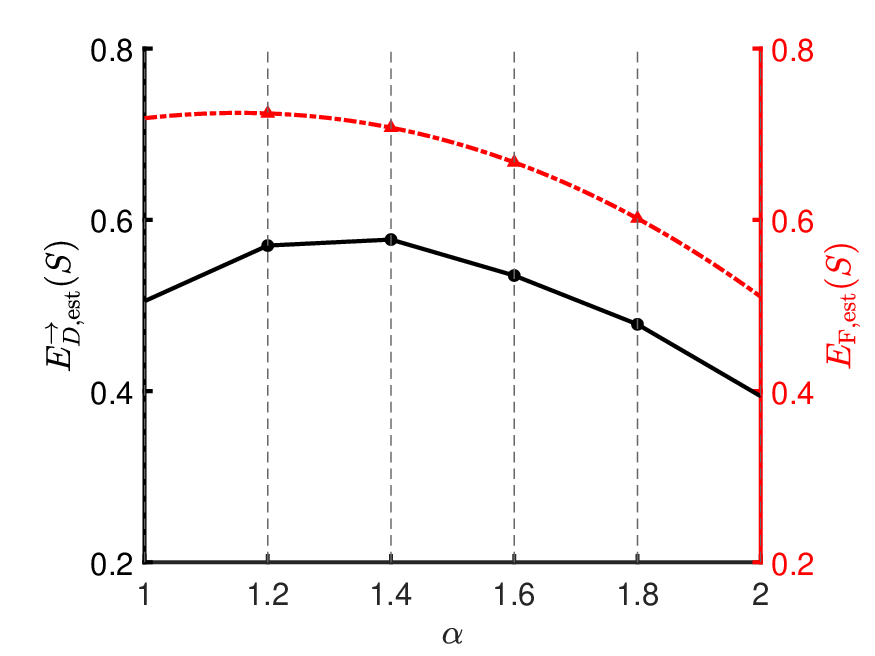}
    \caption{Entanglement estimation results for nonlocal correlations arising from Werner states. The experimental setting is given by $p=0.05,\theta_1=\pi/2,\theta_2=\pi/6$, and $\theta_3=-\pi/6$. We depict the entanglement estimation results using different $\alpha$-CHSH Bell expressions. We plot the estimated values of one-way distillable entanglement and EOF with the black solid line and the red dashed line, respectively.}
\label{fig:werner_entropy_concurrence}
\end{figure}

\begin{theorem}
\label{thm:DI_condition_werner}
    In a Bell test experiment, suppose the underlying state of the system takes the form of Eq.~\eqref{werner}, and the observables take the form of Eq.~\eqref{meas_setting}. For EOF estimation solely from the violation values of $\alpha$-CHSH Bell inequalities, if $\theta_1,\theta_2,\theta_3$ and $p$ satisfy
\begin{equation}
    (1-p)[\sin\theta_1(\sin\theta_2-\sin\theta_3)+\cos\theta_2(\sqrt{2}+1+\cos\theta_1)+\cos\theta_3(\sqrt{2}+1-\cos\theta_1)]>2(1+\sqrt{2}),
    \label{eq:DI_relation_werner}
\end{equation}
then there exists $\alpha>1$, where a better estimation of $E_{\rm F,est}(S)$ can be obtained by using the $\alpha$-CHSH inequality parameterized by this value than by using the original CHSH inequality (corresponding to $\alpha=1$).
\end{theorem}

Theorem~\ref{thm:DI_condition_werner} indicates that the nonlocality depicted by the original CHSH Bell value does not always provide the most accurate EOF estimation of Werner states. In a Bell test experiment, for Werner states parameter $p$ in Eq.~\eqref{werner} and measurement parameters $\theta_1,\theta_2,\theta_3$ in Eq.~\eqref{meas_setting} satisfy Eq.~\eqref{eq:DI_relation_werner}, it is helpful to take a CHSH-type Bell value with $\alpha>1$ to estimate EOF of the Werner state. We leave the proof and discussion of Theorem \ref{thm:DI_condition_werner} in  \ref{appendix:numerical}.

\noindent\textbf{Example.}
As a special example, we take the measurements setting in Eq.~\eqref{eq:DI_condition_pure} with $\theta_1=\pi/2,\theta_3=-\theta_2$, the same one as we use for the case study of non-maximally entangled states. For $0<\theta_2<\pi/4$, any Werner state in Eq.~\eqref{werner} with $p$ satisfying
\begin{equation}
    p<1-\frac{1}{(\sqrt{2}-1)\sin\theta_2+\cos\theta_2}
    \label{DI_condition_werner_p}
\end{equation}
promises a better estimation result of $E_{\rm F,est}$ by using an $\alpha$-CHSH expression with $\alpha>1$ in comparison with $\alpha=1$. The right hand side of Eq.~\eqref{DI_condition_werner_p} is upper bounded by $1-(\sqrt{2}+1)/(\sqrt{4+2\sqrt{2}})$. That is, only a Werner state with $p\leq 1-(\sqrt{2}+1)/(\sqrt{4+2\sqrt{2}})\doteq 0.0761$ is possible to yield the condition in Theorem \ref{thm:DI_condition_werner}. Denote $T=\frac{(1-p)\sin\theta_2}{1-(1-p)\cos\theta_2}$. Then for the underlying system of a Werner state and measurements satisfying Eq.~\eqref{DI_condition_werner_p}, the EOF estimation result reaches its optimal value with parameter $\alpha$
\begin{equation}
    \alpha_E^*=\frac{1}{2}\left(T-\frac{1}{T}\right)>1,
\end{equation}
and the estimation result is
\begin{equation}
    E_{\rm F,est}|_{\alpha_C^*}=1-\frac{p(2-p)}{2[1-(1-p)\cos\theta_2]}.
\end{equation}
Similarly, with an additional assumption on system dimension, we can obtain a more accurate EOF estimation result. We leave the details in  \ref{appendix:numerical}.

\section{Conclusions and discussion}
In this work, we study entanglement estimation via nonlocality, where we consider several entanglement measures for a family of generalized CHSH-type expressions. This family of Bell expressions allows us to effectively reduce the dimension of an unknown system to a pair of qubits, leading to results for particular entanglement measures like negativity, EOF, and one-way distillable entanglement. Under this framework, we also investigate the interplay among entanglement, measurement incompatibility, and nonlocality. While entanglement and measurement incompatibility are both necessary conditions for a nonlocal behavior, under a given nonlocal behavior, their interplay can be subtler than a simple trade-off relation. Given a Bell value, the measurements that require the minimum entanglement are not the most incompatible measurements in general. In addition, we also apply the entanglement estimation results in realistic scenarios. For non-maximally entangled states and Werner states, we analytically show that there exist state and measurements settings where a general CHSH Bell expression with $\alpha>1$ leads to better EOF estimation of the underlying state than the original CHSH expression.

When quantifying entanglement from nonlocality, the estimation results highly depend on the specific entanglement measures. Before our work, there are similar investigations focusing on different entanglement measures~\cite{moroder2013device,toth2015evaluating,arnon2017noise,chen2018exploring,arnon2019device}. A natural question is how nonlocality reflects different entanglement properties. Among various entanglement measures, a notable one is the distillable entanglement. While the necessity of entanglement for unveiling a nonlocal behavior is well-known, whether the underlying state in such a case can always be distilled into the “gold standard” in entanglement resource theory, namely the maximally entangled states, is not clear, which is known as the Peres conjecture~\cite{peres1999all,arnon2021upper}. When considering the entanglement estimation problem, the subtlety resides in that general entanglement distillation allows two-way communication, going beyond the one-way distillable entanglement estimation results in our work and Ref.~\cite{arnon2019device}. In addition, the problem is more involved in a multipartite scenario~\cite{masanes2006asymptotic}. The exact relation between nonlocality and entanglement distillability remains to be explored.

In studying the interplay of entanglement and measurement incompatibility under a given nonlocal behavior, we make some additional assumptions on the measurement observables to ease the quantification of measurement incompatibility. In the sense of a fully DI discussion, one may consider other incompatibility measures, such as the robustness of measurement incompatibility~\cite{designolle2019incompatibility}. Despite the freedom in measuring entanglement and measurement incompatibility, we believe our results unveil the subtlety of the interplay between these nonclassical notions, where more incompatible measurements may not compensate for the absence of entanglement and vice versa. From a resource-theoretic perspective, our results may indicate restrictions on the resource transformation between entanglement and measurement incompatibility in the sense of Bell nonlocality.

When applying our results to experiments, one may consider the practical issues in more detail. For instance, the problem of entanglement estimation via nonlocality can be generalized to the one-shot regime, where one considers dilution and distillation processes with a finite number of possibly non-i.i.d. quantum states. Notably, the results in Ref.~\cite{arnon2019device} provide an approach to estimating one-shot one-way distillable entanglement via nonlocality, and the techniques in Ref.~\cite{buscemi2011entanglement} may be applicable to the estimation of one-shot entanglement cost. We leave research in this direction for future works.

\ack
This work was supported by the National Natural Science Foundation of China Grant No.~12174216 and the Innovation Program for Quantum Science and Technology Grant No.~2021ZD0300804.

Y.Z.~and X.Z.~contributed equally to this work.

\appendix

\section{Reductions of the original optimization problem}
In Sec.~\ref{sc:framework}, we reduce the original optimization problem, including the essential steps of using Jordan's lemma to bypass the dimension problem and reducing a general two-qubit state to the Bell-diagonal state. Here we explain the two steps in detail.

\subsection{Jordan's lemma}\label{appendix:jordan}

Before applying Jordan's lemma to effectively reduce the system dimension, we first note that without loss of generality, the underlying measurements in a Bell test can be taken as projective in a DI analysis. Consider either side of the Bell test, where a measurement setting is characterized by a POVM with elements $\{\hat{M}_i\}_i$ and acts on the state $\rho$. Then, by applying Naimark's dilation theorem~\cite{neumark1943representation}, there exists a quantum channel $F$ that embeds $\rho$ with an ancillary state into a quantum state in some high-dimension Hilbert space, such that for every $\hat{M}_i$, there exists a projector $\hat{V}_i$,
\begin{equation}
    \Tr (\rho \hat{M}_i)=\Tr[\hat{V}_iF(\rho)\hat{V}_i],
\end{equation}
and $\{\hat{V}_i\}_i$ forms a valid projective measurement with $\sum_i \hat{V}_i=\hat{I}$. The dilation only takes a local ancillary state and does not change the measurement statistics. Therefore, in a fully DI setting, one can simply restrict the analysis to projective measurements. In the CHSH-type Bell test, we are dealing with binary observables. Since they can be taken as projective, they can also be described by Hermitian operators with eigenvalues $\pm1$.

We apply Jordan's lemma to bypass the system dimension problem. The description of Jordan's lemma is given below, with the proof can be found in~\cite{Pironio2009device}.

\begin{lemma}
    Suppose $\hat{A}_0$ and $\hat{A}_1$ are two Hermitian operators with eigenvalues $\pm 1$ that act on a Hilbert space with a finite or countable dimension, $\mathcal{H}$. Then there exists a direct-sum decomposition of the system, $\mathcal{H}=\bigoplus\mathcal{H}^{\mu}$, such that $\hat{A}_0=\bigoplus\hat{A}_0^{\mu}$, $\hat{A}_1=\bigoplus\hat{A}_1^{\mu}$, $\hat{A}_0^{\mu},\hat{A}_1^{\mu}\in\mathcal{L}(\mathcal{H}^{\mu})$, where the sub-systems satisfy $\operatorname{dim}\mathcal{H}^{\mu}\leq 2,\forall {\mu}$.
\end{lemma}

Jordan's lemma guarantees that the two possible observables measured by Alice can be represented as
\begin{equation}
  \hat{A}_x=\sum_{\mu}\hat{\Pi}^{\mu_A}\hat{A}_x\hat{\Pi}^{\mu_A}=\bigoplus_{\mu_A}\hat{A}_x^{\mu_A},
\end{equation}
where $x\in\{0,1\}$, $\hat{\Pi}^{\mu_A}$ are projectors onto orthogonal subspaces with dimension no larger than $2$, and $\hat{A}_x^{\mu_A}$ are qubit observables with eigenvalues $\pm1$. A similar representation applies to Bob's measurement observables. Due to the direct-sum representation, one can regard the measurement process as first applying a block-dephasing operation to the underlying quantum system. Consequently, one can equivalently regard the measurement process as measuring the following state,
\begin{equation}
    \bar{\rho}_{AB}=\sum_{\mu}(\hat{\Pi}^{\mu_A}\otimes\hat{\Pi}^{\mu_B})\rho_{AB}(\hat{\Pi}^{\mu_A}\otimes\hat{\Pi}^{\mu_B})=\bigoplus_{\mu}p^{\mu}\rho_{AB}^{\mu}.
\end{equation}
Here we relabel the indices with $\mu\equiv\{\mu_A,\mu_B\}$. As the block-dephasing operators act locally on each side, the measurement process does not increase entanglement in the system. Therefore, we can lower-bound the amount of entanglement in the initial system by studying the average amount of entanglement in the ensemble of qubit-pairs, $\{p^{\mu},\rho_{AB}^{\mu}\}$. 

Consequently, the expected CHSH Bell value in a test is the linear combination of the Bell values for the qubit pairs, $S=\sum_{\mu}p^{\mu}S^{\mu}$. Note that an observer cannot access to the probability distribution, $p^\mu$, and the Bell values for each pair of qubits, $S^\mu$, but only the expected Bell value, $S$, hence the final DI entanglement estimation result should be a function of $S$. On the other hand, we shall first derive entanglement estimation results for each pair of qubits in the form of $E_{\rm est}(S^\mu)$. It is thus essential to consider the convexity of the function, $E_{\rm est}$. If the function is not convex in its argument, i.e., $E_{\rm est}(\sum_\mu p^\mu S^\mu)$ is not smaller than $\sum_\mu p^\mu E_{\rm est}(S^\mu)$, one needs to take the convex closure of $E_{\rm est}$ to obtain a valid lower bound that holds for all possible configurations giving rise to the expected Bell value, $S$.

\subsection{Restriction to Bell-diagonal states}
\label{appendix:Bell-diagonal}
Following the route in Fig.~\ref{fig:flowchart}, the feasible region of the state variables in an entanglement estimation problem can be effectively restricted to the set of Bell-diagonal states on the two qubit systems. We present the following lemma.
\begin{lemma}
    Suppose the underlying system in a CHSH-type Bell test lies in a two-qubit state, $\rho_{AB}$. Then there exists an LOCC that transforms $\rho_{AB}$ into an ensemble of Bell states, $\rho_\lambda$, without changing the expected Bell value.
\end{lemma}
\begin{proof}
     In a CHSH-type Bell test, we transform an arbitrary pair of qubits $\rho_{AB}$ into a Bell-diagonal state $\rho_\lambda$ via three steps of LOCC. In each step, We verify that the $\alpha$-CHSH Bell values are equal for the states before and after the transformation with the same measurements.
    
    \textit{Step 1:} 
    In a CHSH Bell test, Alice and Bob fix their local computational bases, or, the axes of the Bloch spheres on each side. As there are only two observables on each side, one can represent them on the $x-z$ plane of the Bloch sphere without loss of generality. Then Alice and Bob flip their measurement results simultaneously via classical communication with probability $1/2$. This operation can be interpreted as transforming $\rho_{AB}$ into the following state,
    \begin{equation}
        \rho_1=\frac{1}{2}[\rho_{AB}+(\sigma_2\otimes\sigma_2)\rho_{AB}(\sigma_2\otimes\sigma_2)].
    \end{equation}
    To avoid confusion about the subscripts, we use the following convention to denote the Pauli operators in the Appendix,
    \begin{equation}
    \begin{split}
        \sigma_x&\equiv\sigma_1, \\
        \sigma_y&\equiv\sigma_2, \\
        \sigma_z&\equiv\sigma_3.
    \end{split}
    \end{equation}
    Under the Bell basis determined by the local computational bases, $\{\ket{\Phi^+},\ket{\Psi^-},\ket{\Phi^-},\ket{\Psi^+}\}$, $\rho_1$ can be denoted as
    \begin{equation}
    \rho_1=\begin{bmatrix} 
        \lambda_{\Phi^+} &l_1e^{i\phi_1} &0 &0 \\
        l_1e^{-i\phi_1} &\lambda_{\Psi^-} &0 &0 \\
        0 &0 &\lambda_{\Psi^+} &l_2e^{i\phi_2} \\
        0 &0 &l_2e^{-i\phi_2} &\lambda_{\Phi^-}
        \end{bmatrix},
    \end{equation}
where $\lambda_i,i=\Phi^\pm,\Psi^\pm$ are  eigenvalues of the corresponding Bell state and $l_1,l_2,\phi_1,\phi_2$ are off-diagonal parameters. It can be verified that the statistics of measuring $\sigma_i\otimes\sigma_j$ for $i,j=1,3$ are invariant under the operation of $\sigma_2\otimes\sigma_2$. Thus 
\begin{equation}
    \Tr [(\sigma_2\otimes\sigma_2)\rho_{AB}(\sigma_2\otimes\sigma_2)(\hat{A}_x\otimes\hat{B}_y)]=\Tr [\rho_{AB}(\hat{A}_x\otimes\hat{B}_y)]
\end{equation}
for $\hat{A}_x$ and $\hat{B}_y$, $x,y=0,1$, which indicates that $\rho_1$ and $\rho_{AB}$ share the common Bell value.

\textit{Step 2:} In this step, we apply LOCC to transform $\rho_1$ into a state where the off-diagonal terms on the Bell basis become imaginary numbers. For this purpose, Alice and Bob can each apply a local rotation around the $y$-axes of the Bloch spheres on their own systems for an angle $\theta$ by
\begin{equation}
    R_y(\theta)=\cos{\frac{\theta}{2}}I+i\sin{\frac{\theta}{2}}\sigma_2=\left(
    \begin{smallmatrix}
    \cos{\frac{\theta}{2}} & \sin{\frac{\theta}{2}} \\
    -\sin{\frac{\theta}{2}} & \cos{\frac{\theta}{2}}
    \end{smallmatrix}
    \right),
\end{equation}
with its action on a general observable characterized by $\gamma$ residing in the $x-z$ plane given by
    \begin{equation}
        R_y(\theta)(\cos\gamma\sigma_1+\sin\gamma\sigma_3)=\cos(\gamma+\frac{\theta}{2})\sigma_1+\sin(\gamma+\frac{\theta}{2})\sigma_3.
        \label{rotation}
    \end{equation}
After applying the operation, the resulting state becomes $\rho_2=[R_y(\alpha)\otimes R_y(\beta)]\rho_1[R_y(-\alpha)\otimes R_y(-\beta)]$, where the off-diagonal terms undergo the following transformations,
\begin{gather}
    l_1e^{i\phi_1}\rightarrow\frac{1}{2}(\lambda_{\Phi^+}-\lambda_{\Psi^-})\sin(\alpha-\beta)+l_1\cos\phi_1\cos(\alpha-\beta)+l_1\sin\phi_1 i,\label{terms_rotation1}\\
    l_2e^{i\phi_2}\rightarrow\frac{1}{2}(\lambda_{\Phi^-}-\lambda_{\Psi^+})\sin(\alpha+\beta)+l_2\cos\phi_2\cos(\alpha+\beta)+l_2\sin\phi_2 i.\label{terms_rotation2}
\end{gather}
By choosing $\alpha$ and $\beta$ properly, the real parts in the off-diagonal terms of $\rho_2$ can be eliminated. Similarly, as in the \textit{Step 1}, the measurement of $\sigma_i\otimes\sigma_j$ for $i,j=1,3$ remains invariant under local rotations around the $y$-axes, which indicates $\rho_2$ and $\rho_1$ give the same Bell value under the same measurements.

\textit{Step 3:}
Note that $\rho_2$ and $\rho_2^*$ give the same Bell value under the given measurements,
\begin{equation}
    \Tr[\rho_2(\sigma_i\otimes\sigma_j)]=\Tr[\rho_2^*(\sigma_i\otimes\sigma_j)],i,j=1,3.
\end{equation}
Hence without loss of generality, one can take the underlying state in the Bell test as $\rho_{\lambda}=(\rho_2+\rho_2^*)/2$, which is a Bell-diagonal state. 
\end{proof}

Based on the above simplification, we represent Eq.~\eqref{equivalent_optm} under Bell-diagonal states, which leads to the following lemma.
\begin{lemma}
The maximal value of the $\alpha$-CHSH expression in Eq.~\eqref{CHSH_type} for a Bell-diagonal state shown in Eq.~\eqref{Bell_diagonal}, $\rho_{\lambda}$, is given by
\begin{equation}
    S= 2\sqrt{\alpha^2(\lambda_{1}+\lambda_{2}-\lambda_{3}-\lambda_{4})^2+(\lambda_{1}-\lambda_{2}+\lambda_{3}-\lambda_{4})^2},
    \label{eq:maximal_value_appendix}
\end{equation}
where $\lambda_i$ is the $i$-th largest eigenvalue of $\rho_{\lambda}$.
\end{lemma}
\begin{proof}
In an $\alpha$-CHSH Bell test, measurements corresponding to non-degenerate Pauli observables can be expressed as
\begin{gather}
        \hat{A_x}=\Vec{a}_x\cdot\Vec{\sigma},\\
        \hat{B_y}=\Vec{b}_y\cdot\Vec{\sigma},
\end{gather}
where $\Vec{\sigma}=(\sigma_1,\sigma_2,\sigma_3)$ are three Pauli matrices, and $\Vec{a}_x=(a_x^1,a_x^2,a_x^3)$ and $\Vec{b}_y=(b_y^1,b_y^2,b_y^3)$ are unit vectors for $x,y=0,1$. A Bell-diagonal state shown in Eq.~\eqref{Bell_diagonal} can be expressed on the Hilbert-Schmidt basis as
\begin{equation}
    \rho_{\lambda}=\frac{1}{4}\left(I+\sum_{i,j=1}^3 T_{\lambda,ij}\sigma_i\otimes\sigma_j\right),
\end{equation}
where 
\begin{equation}
    T_{\lambda}=
    \begin{bmatrix} 
        (\lambda_1 + \lambda_3)-(\lambda_2 +\lambda_4) &0 &0  \\
        0 &(\lambda_3 + \lambda_2)-(\lambda_1 + \lambda_4) &0 \\
        0 &0 &(\lambda_1 + \lambda_2)-(\lambda_3 + \lambda_4)
    \end{bmatrix}
    \label{T_lambda}
\end{equation}
is a diagonal matrix. The $\alpha$-CHSH expression in Eq.~\eqref{CHSH_type} can be expressed in terms of $T_{\lambda}$ as
\begin{equation}
    \begin{aligned}
        &\Tr\left\{\alpha\rho_{\lambda}(\Vec{a}_0\cdot\Vec{\sigma})\otimes[(\Vec{b}_0+\Vec{b}_1)\cdot\Vec{\sigma}]+\rho_{\lambda}(\Vec{a}_1\cdot\sigma)\otimes[(\Vec{b}_0-\Vec{b}_1)\cdot\Vec{\sigma}]\right\}\\
        =&\alpha[\Vec{a}_0\cdot T_{\lambda}(\Vec{b}_0+\Vec{b}_1)]+[\Vec{a}_1\cdot T_{\lambda}(\Vec{b}_0-\Vec{b}_1)].
    \end{aligned}
\end{equation}
Following the method in Ref.~\cite{horodecki1995violating}, we introduce a pair of normalized orthogonal vectors, $\Vec{c}_0$ and $\Vec{c}_1$,
\begin{gather}
    \Vec{b}_0+\Vec{b}_1=2\cos\theta\Vec{c}_0,\\
    \Vec{b}_0-\Vec{b}_1=2\sin\theta\Vec{c}_1
\end{gather}
where $\theta\in[0,\pi/2]$. This gives the maximal $\alpha$-CHSH Bell value,
\begin{equation}
    \begin{aligned}
        S
        &=\max_{\Vec{a}_0,\Vec{a}_1,\Vec{c}_0,\Vec{c}_1,\theta}2\alpha\cos\theta(\Vec{a}_0\cdot T_{\lambda}\Vec{c}_0)+2\sin\theta(\Vec{a}_1\cdot T_{\lambda}\Vec{c}_1).
    \end{aligned}
\end{equation}
The maximization of the Bell value is taken over parameters $\Vec{a}_x,\Vec{b}_y$ for $x,y=0,1$, with the parameters $\lambda_{i}$ fixed. We obtain
\begin{equation}
    \begin{aligned}
        S&=\max_{\Vec{c}_0,\Vec{c}_1,\theta}2\alpha\cos\theta|T_{\lambda}\Vec{c}_0|+2\sin\theta|T_{\lambda}\Vec{c}_1|\\
        &=\max_{\Vec{c}_0,\Vec{c}_1}2\sqrt{\alpha^2|T_{\lambda}\Vec{c}_0|^2+|T_{\lambda}\Vec{c}_1|^2},\label{max_CHSH_vio_Bell_c}
    \end{aligned}
\end{equation}
where the first equality in Eq.~\eqref{max_CHSH_vio_Bell_c} is saturated when $\Vec{a}_x=T_{\lambda}\Vec{c}_x/|T_{\lambda}\Vec{c}_x|, x=0,1$, and the second inequality is saturated when $\tan\theta=|T_{\lambda}\Vec{c}_1|/(\alpha|T_{\lambda}\Vec{c}_0|)$. Since $\alpha>1$ and $\Vec{c}_0$ and $\Vec{c}_1$ are orthonormal vectors, the maximum of the second line in Eq.~\eqref{max_CHSH_vio_Bell_c} is obtained when $|T_{\lambda}\Vec{c}_0|$ and $|T_{\lambda}\Vec{c}_1|$ equal to the absolute values of the largest and the second largest eigenvalues of $T_{\lambda}$, respectively. Without loss of generality, we assume $\lambda_1\geq\lambda_2\geq\lambda_3\geq\lambda_4$ in $\rho_{\lambda}$. This leads to the ordering of the absolute values of the elements of $T_{\lambda}$,
\begin{equation}
\begin{split}
    |T_{\lambda,33}|&=|(\lambda_{1}-\lambda_{4})+(\lambda_{2}-\lambda_{3})|\geq|(\lambda_{1}-\lambda_{4})-(\lambda_{2}-\lambda_{3})|=|T_{\lambda,11}|,\\
    |T_{\lambda,11}|&=|(\lambda_{3}-\lambda_{4})+(\lambda_{1}-\lambda_{2})|\geq|(\lambda_{3}-\lambda_{4})-(\lambda_{1}-\lambda_{2})|=|T_{\lambda,22}|.
\end{split}
\end{equation}
Thus, the second line in Eq.~\eqref{max_CHSH_vio_Bell_c} reaches its maximum when $|\Vec{c}_0|=(0,0,1)^T$ and $|\Vec{c}_1|=(1,0,0)^T$. Therefore, for any given Bell-diagonal state $\rho_{\lambda}$ in Eq.~\eqref{Bell_diagonal}, the maximal $\alpha$-CHSH Bell value is
\begin{equation}
    S= 2\sqrt{\alpha^2(\lambda_{1}+\lambda_{2}-\lambda_{3}-\lambda_{4})^2+(\lambda_{1}-\lambda_{2}+\lambda_{3}-\lambda_{4})^2},
\end{equation}
where measurements for $\rho_{\lambda}$ to achieve the maximal Bell value, i.e., optimal measurements, are given by
\begin{equation}
    \begin{gathered}
        \hat{A}_0=\pm\sigma_z,\\
        \hat{A}_1=\pm\sigma_x,\\
        \hat{B}_0=\pm\cos\theta\sigma_z\pm\sin\theta\sigma_x,\\
        \hat{B}_1=\pm\cos\theta\sigma_z\mp\sin\theta\sigma_x,
    \end{gathered}
    \label{opt_meas1}
\end{equation}
    or
\begin{equation}
    \begin{gathered}
        \hat{A}_0=\pm\sigma_z,\\
        \hat{A}_1=\mp\sigma_x,\\
        \hat{B}_0=\pm\cos\theta\sigma_z\mp\sin\theta\sigma_x,\\
        \hat{B}_1=\pm\cos\theta\sigma_z\pm\sin\theta\sigma_x,
    \end{gathered}
    \label{opt_meas2}
\end{equation}
with $\tan\theta=(\lambda_{1}-\lambda_{2}+\lambda_{3}-\lambda_{4})/[\alpha(\lambda_{1}+\lambda_{2}-\lambda_{3}-\lambda_{4})]$.
\end{proof}

From the proof, we see that any Bell-diagonal state $\rho_{\lambda}$ in Eq.~\eqref{Bell_diagonal} reaches its maximal Bell value of Eq.~\eqref{CHSH_type} when measurements are taken in the form Eq.~\eqref{opt_meas1} or Eq.~\eqref{opt_meas2}. In other words, measurements in Eq.~\eqref{opt_meas1} and Eq.~\eqref{opt_meas2} are the optimal measurements for $\rho_{\lambda}$ that yield the largest $\alpha$-CHSH Bell value. To solve the simplified entanglement estimation problem in Eq.~\eqref{equivalent_optm} for Bell-diagonal states, we need to solve the optimal measurements first. Given a general pair of qubits $\rho_{AB}$, the maximal $\alpha$-CHSH Bell value $S$ for $\rho_{AB}$ is expressed as a function of $T_{ij}$,
\begin{equation}
    S=[2(\alpha^2+1)(T_{11}^2+T_{13}^2+T_{31}^2+T_{33}^2)+2(\alpha^2-1)\sqrt{(T_{11}^2-T_{13}^2+T_{31}^2-T_{33}^2)^2+4(T_{11}+T_{13}+T_{31}+T_{33})^2}]^{1/2},
\label{max_CHSH_vio_gen}
\end{equation}
where $T_{ij}=\Tr[\rho_{AB}(\sigma_i\otimes\sigma_j)]$ is the coefficient of $\rho_{AB}$ under Hilbert-Schmidt basis. When $\rho_{AB}$ is Bell-diagonal, Eq.~\eqref{max_CHSH_vio_gen} degenerates to Eq.~\eqref{eq:maximal_value_appendix}.

\section{Proof of the lower bound of concurrence}
\label{appendix:concurrence_est}
In this section, we prove the analytical concurrence estimation result via the $\alpha$-CHSH Bell value. Here we restrict the underlying state as a pair of qubits.

\begin{theorem}
    Suppose the underlying quantum state is a pair of qubits. For a given $\alpha$-CHSH expression in Eq.~\eqref{CHSH_type} parametrized by $\alpha$, if the Bell value is $S$, then the amount of concurrence in the underlying state can be lower-bounded,
\begin{equation}
    C(\rho_{AB})\geq \sqrt{\frac{S^2}{4}-\alpha^2}.
\end{equation}
The equality can be saturated when measuring a Bell-diagonal state in Eq.~\eqref{Bell_diagonal} with eigenvalues
\begin{equation}
    \begin{gathered}
        \lambda_1=\frac{1}{2}+\frac{1}{2}\sqrt{\frac{S^2}{4}-\alpha^2},\\
        \lambda_2=\frac{1}{2}-\frac{1}{2}\sqrt{\frac{S^2}{4}-\alpha^2},\\
        \lambda_3=\lambda_4=0,
    \end{gathered}
\end{equation}
using measurements in Eq.~\eqref{Bell_opt_meas} with $\theta=\arctan(\frac{1}{\alpha}\sqrt{\frac{S^2}{4}-\alpha^2})$.
\end{theorem}

\begin{proof}
Given any Bell value  $S\in(2\alpha,2\sqrt{1+\alpha^2}]$, we aim to determine the least amount of concurrence that is required to support the Bell value, $S$. We solve the simplified optimization problem in Eq.~\eqref{equivalent_optm}, restricting the underlying state as a Bell-diagonal state in Eq.~\eqref{Bell_diagonal} and taking the objective entanglement measure as $C(\cdot)$,
\begin{equation}
    \begin{split}
    C_{\rm est} & =\underset{\lambda_i,i=1,2,3,4}{\text{min}} \max\{0,2\lambda_{1}-1\}, \\
    \text{s.t.}\quad
    S &=2\sqrt{\alpha^2(\lambda_1+\lambda_2-\lambda_3-\lambda_4)^2+(\lambda_1-\lambda_2+\lambda_3-\lambda_4)^2},\\
    \lambda_1 & \geq\lambda_2\geq\lambda_3\geq\lambda_4,\\
    1 & =\sum_{i=1}^4\lambda_i, \lambda_i\geq 0 \; ,i=1,2,3,4.
    \end{split}
    \label{concurrence_optm}
\end{equation}
We first reduce the number of variables to simplify the optimization in Eq.~\eqref{concurrence_optm}. Since the variables in Eq.~\eqref{concurrence_optm} are not independent with each other, we express variables $\lambda_1$ and $\lambda_4$ as functions of variables $\lambda_2$ and $\lambda_3$, 
\begin{gather}
\lambda_{\max}=\lambda_{1}=\frac{1}{2}-\frac{1}{\alpha^2+1}(\alpha^2\lambda_{2}+\lambda_{3})+\frac{1}{\alpha^2+1}\sqrt{\frac{S^2(\alpha^2+1)}{16}-\alpha^2(\lambda_{2}-\lambda_{3})^2},\label{lambda_max}\\
\lambda_{\min}=\lambda_{4}=\frac{1}{2}-\frac{1}{\alpha^2+1}(\lambda_{2}+\alpha^2\lambda_{3})-\frac{1}{\alpha^2+1}\sqrt{\frac{S^2(\alpha^2+1)}{16}-\alpha^2(\lambda_{2}-\lambda_{3})^2}\label{lambda_min}.
\end{gather}
The non-negativity of $\lambda_{\min}$ in Eq.~\eqref{lambda_min} restricts $\lambda_{2}$ and $\lambda_{3}$ outside an ellipse,
\begin{equation}
    \left(\lambda_{2}-\frac{1}{2}\right)^2+\alpha^2\left(\lambda_{3}-\frac{1}{2}\right)^2\geq\frac{S^2}{16},
\end{equation}
and the fact that $\lambda_{4}$ in Eq.~\eqref{lambda_min} is the smallest among all $\lambda_i$ restricts $\lambda_{2}$ and $\lambda_{3}$ inside an ellipse,
\begin{equation}
    \lambda_{2}^2+(4\alpha^2+1)\lambda_{3}+2\lambda_{2}\lambda_{3}-\lambda_{2}-(2\alpha^2+1)\lambda_{3}\leq\frac{S^2}{16}-\frac{\alpha^2+1}{4}.
\end{equation}
With the above derivations, the optimization in Eq.~\eqref{concurrence_optm} can be rewritten with independent variables $\lambda_{2}$ and $\lambda_{3}$ as

\begin{equation}
    \begin{split}
    C_{\rm est} & =\underset{\lambda_{2},\lambda_{3}}{\text{min}} 2\lambda_{1}-1, \\
    \text{s.t.}\quad
    \lambda_{1} &=\frac{1}{2}-\frac{1}{\alpha^2+1}(\alpha^2\lambda_{2}+\lambda_{3})+\frac{1}{\alpha^2+1}\sqrt{\frac{S^2(\alpha^2+1)}{16}-\alpha^2(\lambda_{2}-\lambda_{3})^2},\\
    0 & \leq(\lambda_{2}-\frac{1}{2})^2+\alpha^2(\lambda_{3}-\frac{1}{2})^2-\frac{S^2}{16},\\
    0 & \geq \lambda_{2}^2+(4\alpha^2+1)\lambda_{3}+2\lambda_{2}\lambda_{3}-\lambda_{2}-(2\alpha^2+1)\lambda_{3}-(\frac{S^2}{16}-\frac{\alpha^2+1}{4}),\\
    \lambda_{1}&\geq\lambda_{2}\geq\lambda_{3},\\
    0&\leq \lambda_{2},\lambda_{3},\lambda_{2}+\lambda_{3}\leq 1.
    \end{split}
    \label{concurrence_optm_solvable}
\end{equation}
The optimization in Eq.~\eqref{concurrence_optm_solvable} can be solved analytically, with the global optimal value taken at
\begin{equation}
    \begin{gathered}
        \lambda_1=\frac{1}{2}+\frac{1}{2}\sqrt{\frac{S^2}{4}-\alpha^2},\\
        \lambda_2=\frac{1}{2}-\frac{1}{2}\sqrt{\frac{S^2}{4}-\alpha^2},\\
        \lambda_3=\lambda_4=0.
        \label{optimal_state}
    \end{gathered}
\end{equation}
Therefore the estimated concurrence is lower-bounded,
\begin{equation}
    C(\rho_{AB})\geq C_{\rm est}(S)= \sqrt{\frac{S^2}{4}-\alpha^2}.
    \label{concurrence_est_appendix}
\end{equation}
The lower bound of Eq.~\eqref{concurrence_est_appendix} is saturated when the $\alpha$-CHSH Bell value $S$ is obtained by measuring the Bell-diagonal state $\rho_{\lambda}$ with the parameters in Eq.~\eqref{optimal_state} under its optimal measurements. The optimal measurements are in Eq.~\eqref{opt_meas1} and Eq.~\eqref{opt_meas2} with $\theta=\arctan(\frac{1}{\alpha}\sqrt{\frac{S^2}{4}-\alpha^2})$.
\end{proof}

\section{Entanglement of formation estimation result in Ref.~\cite{arnon2017noise}}
\label{sc:EOF_arnon}
In Ref.~\cite{arnon2017noise}, the authors use the rigidity property of nonlocal games to estimate the EOF in the underlying system. Here, we briefly review their results. The Bell test can be formulated as a nonlocal game, $G$. Under this description, the nonlocal parties win the nonlocal game if their measurement outputs satisfy a set of relations with respect to their input settings that correspond to the Bell test~\cite{elkouss2016nearly}. For instance, in the nonlocal game related to the CHSH Bell test, Alice and Bob win the game if $ab=(-1)^{xy}$. We denote the nonlocal parties' maximum winning probability with classical strategies as $\mathrm{cval}(G)$ and their maximum winning probability with quantum strategies as $\mathrm{qval}(G)$. In the CHSH game, $\mathrm{qval}(G)=(2+\sqrt{2})/4$, corresponding to the Tsirelson bound, $S=2\sqrt{2}$, and $\mathrm{cval}(G)=3/4$, corresponding to the maximum Bell value subjected to local hidden variables, $S=2$. We denote the gap between classical and quantum strategies as $\Delta=\textrm{qval}(G)-\textrm{cval}(G)$. Given the nonlocal game $G$ with a non-zero gap $\Delta$, suppose the nonlocal parties play the game for $n$ independent instances in parallel. For a given value, $\nu\in[0,\mathrm{qval}(G)]$, we say the nonlocal parties win a threshold game, $G_{\mathrm{qval}(G)-\nu}^n$, if they win a fraction of at least $[\mathrm{qval}(G)-\nu]$ instances in total. Suppose there exists a quantum strategy for the nonlocal parties such that they win $G_{\mathrm{qval}(G)-\nu}^n$ with probability $\kappa$. Then, the EOF in the underlying system of the threshold game can be lower-bounded,
\begin{equation}
    E_{\mathrm{F}}(\rho)\geq c_2\kappa^2 n,
    \label{EOF_2017}
\end{equation}
where
\begin{equation}
    c_2=\frac{(\Delta-\nu)^5}{10\cdot 180^2\log |\mathcal{A}\times\mathcal{B}|},
\end{equation}
with $\mathcal{A}$ and $\mathcal{B}$ denoting the answer alphabets of Alice and Bob, respectively.

We compare this result with ours when using the CHSH Bell test. As we estimate the EOF in the Shannon limit with the i.i.d. assumption, for a fair comparison, we take $n\rightarrow\infty$ and evaluate the average EOF per instance in Eq.~\eqref{EOF_2017}. When the Bell value is $S>2$, Alice can Bob can win the corresponding nonlocal game $G$ with probability $1/2+S/8$. Applying this strategy to each instance in the threshold game in an i.i.d. manner, Alice and Bob can win the threshold game, $G_{\mathrm{qval}(G)-\nu}^{n\rightarrow\infty}$ with $\nu=\sqrt{2}/4-S/8$, with probability $\kappa=1$. Taking these values into Eq.~\eqref{EOF_2017}, the EOF of the system per instance is lower-bounded by
\begin{equation}
\begin{split}
    E_{\mathrm{F}}(\rho_{AB})&\geq \frac{[(\frac{2+\sqrt{2}}{4}-\frac{3}{4})-(\frac{\sqrt{2}}{4}-\frac{S}{8})]^5}{10\cdot 180^2\cdot\log (2\times 2)}\\
    &=\frac{(S-2)^5}{10\cdot 180^2\cdot2^{16}}.
\end{split}
\end{equation}

\section{Realistic settings in experiment}
\label{appendix:numerical}
In this section, we analyze the realistic settings and analytically derive the condition where a better concurrence estimation can be derived with a tilted CHSH Bell expression. That is, by using the family of $\alpha$-CHSH expressions, a Bell expression with parameter $\alpha>1$ gives a better estimation result than $\alpha=1$. Considering the estimation function, $C_{\rm est}(S)$, as a function parameterized by $\alpha$, then our target is to determine the condition for the following inequalities,
\begin{equation}
\begin{gathered}
    \frac{\partial C_{\rm est}(S)}{\partial \alpha}|_{\alpha=1}>0,\\
    C_{\rm est}(S)|_{\alpha=1}>0.
\end{gathered}
\label{increasing_estimation_condition}
\end{equation}
The conclusions of Theorem~\ref{thm:DI_condition_pure} and Theorem~\ref{thm:DI_condition_werner} can be directly solved from  Eq.~\eqref{increasing_estimation_condition} by substituting the corresponding estimation equation and the $\alpha$-CHSH Bell value.

For a better understanding of Theorem~\ref{thm:DI_condition_pure}, we take $\theta_1=\pi/2$ in Eq.~\eqref{eq:DI_condition_pure}. With straightforward derivations, we find that when the measurement parameters, $\theta_2$ and $\theta_3$, satisfy
\begin{equation}
    (\sqrt{2}+1)(\cos\theta_2+\cos\theta_3)+(\sin\theta_2-\sin\theta_3)>2(\sqrt{2}+1),
    \label{DI_condition_pure_werner}
\end{equation}
there exists a value of $\delta$ such that $\theta_1=\pi/2$, and $\theta_2,\theta_3$ and $\delta$ satisfy the inequality in  Eq.~\eqref{eq:DI_condition_pure}. In other words, when measurement parameters are set in Eq.~\eqref{meas_setting} with $\theta_1=\pi/2$ and $\theta_2,\theta_3$ follow Eq.~\eqref{DI_condition_pure_werner}, there exists a proper state, $\ket{\phi_{AB}(\delta)}$, such that concurrence estimation result $C_{\rm est}(S)$ of the $\ket{\phi_{AB}(\delta)}$ for some $\alpha>1$ is larger than that with $\alpha=1$. The conclusion from Eq.~\eqref{DI_condition_pure_werner} also applies to Werner states.

\section{Semi-device-independent optimal entanglement estimation}
In some scenarios, one may trust the functioning of the source, such as knowing the input states to be pairs of qubits, which can be seen as a semi-DI scenario. With the additional information, one may obtain a better entanglement estimation result. 
In this section, we compare the performance of this semi-DI scenario with the fully DI scenario under a realistic setting. Suppose the underlying state is a non-maximally entangled state $\ket{\phi_{AB}(\delta)}$ in Eq.~\eqref{pure} with $\delta=0.6$, and the measurements are given by Eq.~\eqref{meas_setting} with $\theta_1=\pi/2$ and $\theta_2=-\theta_3=\pi/2-1.2$. Under this setting, the fully DI and semi-DI EOF estimation results are illustrated in Fig.~\ref{fig:DI_semi-DI_pure_concurrence}. The semi-DI estimation, $E_{\rm F,est,semi-DI}(S)$, is strictly larger than the fully DI estimation, $E_{\rm F,est,DI}(S)$, for any value of $\alpha>1$. Besides, when $\alpha$ takes the value of $\alpha^*_{E}\doteq 2.3973$, $E_{\rm F,est,semi-DI}|_{\alpha^*_E}\doteq 0.9031$ rigorously equals to the real system EOF when $\delta=0.6$.

\begin{figure}[hbt!]
    \centering    \includegraphics[scale=0.57]{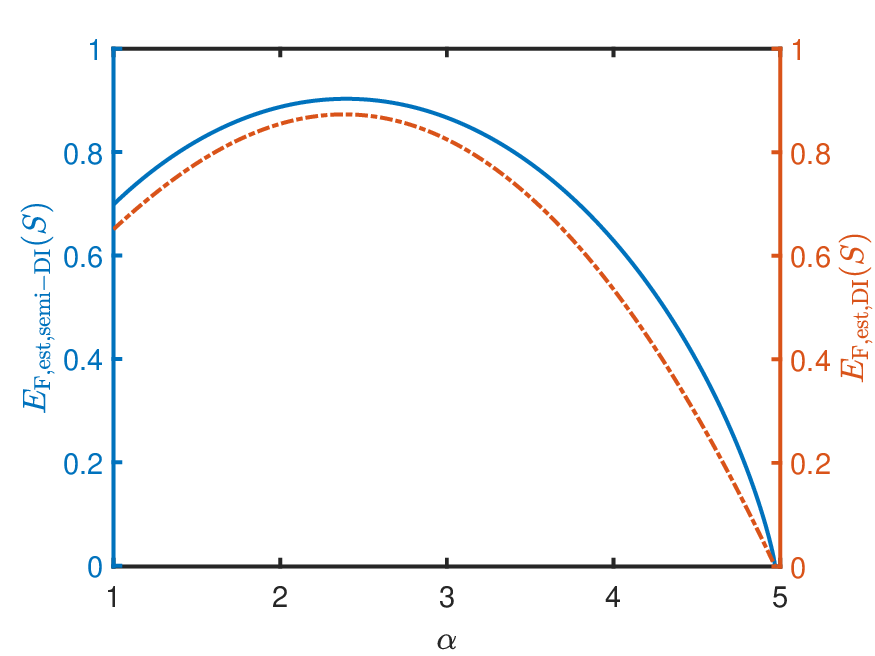}
    \caption{Illustration of the comparison between DI and semi-DI EOF estimation results. The experimental setting is given by $\delta=0.6,\theta_1=\pi/2,\theta_2=-\theta_3=\pi/2-1.2$. For the estimation results varying in $\alpha>1$, we plot the estimated values from the settings of semi-DI and DI with the blue solid line and the red dashed line, respectively. The semi-DI EOF estimation with confirmed knowledge of input dimensions is strictly larger than the DI EOF estimation.
    }
    \label{fig:DI_semi-DI_pure_concurrence}
\end{figure}

In addition, we also study the condition where a better entanglement estimation result using general $\alpha$-CHSH expressions is obtained under some $\alpha>1$. We present the following theorems for the state families of Werner states and non-maximally entangled states, using concurrence as the entanglement measure.

\begin{theorem}
\label{thm:semi-DI_condition_pure}
In a Bell test experiment where the input states are pairs of qubits, suppose the underlying state of the system takes the form of Eq.~\eqref{pure} and the observables take the form of Eq.~\eqref{meas_setting}. When $\theta_1,\theta_2,\theta_3$ and $\delta$ satisfy
\begin{equation}
    (\cos\theta_2+\cos\theta_3)[\sin2\delta\sin\theta_1(\sin\theta_2-\sin\theta_3)+(1+\cos\theta_1)\cos\theta_2+(1-\cos\theta_1)\cos\theta_3]> 4,
    \label{eq:semi-DI_condition_pure}
\end{equation}
there exists $\alpha>1$, where a better estimation of $C_{\rm est}(S)$ can be obtained by using the $\alpha$-CHSH inequality parameterized by this value than by using the original CHSH inequality (corresponding to $\alpha=1$).
\end{theorem}

The proof of Theorem~\ref{thm:semi-DI_condition_pure} is similar to the proof of Theorem~\ref{thm:DI_condition_pure}. Here, we alternatively apply the concurrence estimation result for pairs of qubits input, $C_{\rm est}(S)$ in Eq.~\eqref{concurrence_est}, to the condition in Eq.~\eqref{increasing_estimation_condition}. To better understand the theorem, we present an example with $\theta_1=\pi/2$ in Eq.~\eqref{eq:semi-DI_condition_pure}. When
\begin{equation}
\begin{split}
  \theta_2+\theta_3&=\arccos(\frac{4}{1+\sqrt{2}k}-1),\\
  -\frac{\pi}{4}-\arccos k&<\theta_3-\theta_2<-\frac{\pi}{4}+\arccos k,
\end{split}
\label{semi-DI_condition_pure_werner_k}
\end{equation}
where $k\in[\sqrt{2}/2,1)$, there exists $\delta$ such that $\theta_1=\pi/2,\theta_2,\theta_3$ and the $\delta$ satisfy Eq.~\eqref{eq:semi-DI_condition_pure}. In other words, when measurement parameters are set in Eq.~\eqref{meas_setting} with $\theta_1=\pi/2$ and $\theta_2,\theta_3$ following Eq.~\eqref{semi-DI_condition_pure_werner_k}, there exist non-maximally entangled states where a better semi-DI concurrence estimation result is obtained for some $\alpha>1$ in comparison with $\alpha=1$.

In Fig.~\ref{fig:DI_semi-DI_pure_concurrence}, we observe that under a well-chosen value of $\alpha$, the semi-DI concurrence estimation coincides with the real value. In many semi-DI CHSH Bell tests, the existence of $\alpha$ that yields an accurate estimation of state concurrence is ubiquitous. Theorem~\ref{thm:concurrence_est} indicates that under the assumption of qubit inputs, the lower bound of concurrence in Eq.~\eqref{concurrence_est} can be saturated at any non-maximally entangled state $\ket{\phi_{AB}(\delta)}$, once the Bell value is obtained by the optimal measurements of the $\ket{\phi_{AB}(\delta)}$. In fact, earlier research indicates that the observables,
 \begin{equation}
  \begin{split}
      \hat{A}_0&=\pm\sigma_z,\\
    \hat{A}_1&=\sigma_x, \\
    \hat{B}_0&=\pm\cos\theta\sigma_z+\sin\theta\sigma_x,\\
    \hat{B}_1&=\pm\cos\theta\sigma_z-\sin\theta\sigma_x,
  \end{split}
  \label{optimal_meas_pure}
 \end{equation}
with $\tan\theta=\sin2\delta/\alpha$, are the optimal measurements of $\ket{\phi_{AB}}$ with any fixed $\alpha$~\cite{acin2012randomness} . The form of observables in Eq.~\eqref{optimal_meas_pure} coincides with our initialization in Eq.~\eqref{meas_setting} when  $\theta_1=\pi/2,\theta_2+\theta_3=0,0<\theta_2<\pi/4$. In this case, any non-maximally entangled state $\ket{\phi_{AB}(\delta)}$ with $\delta$ satisfying $\sin2\delta>\tan\theta_2$ reaches its optimal concurrence estimation when $\alpha$ takes the value of
\begin{equation}
\alpha^*_C=\frac{\sin2\delta}{\tan\theta_2},
\end{equation}
and the optimal semi-DI estimation value is
\begin{equation}
    C_{\rm est}|_{\alpha^*_C,\rm semi-DI}=C(\ket{\phi_{AB}}).
\end{equation}

\begin{theorem}
\label{thm:semi-DI_condition_werner}
In a Bell test experiment where the input states are pairs of qubits, suppose the underlying state of the system takes the form of Eq.~\eqref{werner} and the observables take the form of Eq.~\eqref{meas_setting}. When $\theta_1,\theta_2,\theta_3$ and $p$ satisfy
\begin{equation}
    (1-p)^2(\cos\theta_2+\cos\theta_3)\cdot[\sin\theta_1(\sin\theta_2-\sin\theta_3)+(1+\cos\theta_1)\cos\theta_2+(1-\cos\theta_1)\cos\theta_3]>4,
    \label{eq:semi-DI_condition_werner}
\end{equation}
there exists $\alpha>1$, where a better estimation of $C_{\rm est}(S)$ can be obtained by using the $\alpha$-CHSH inequality parameterized by this value than by using the original CHSH inequality (corresponding to $\alpha=1$).
\end{theorem}

The proof of Theorem~\ref{thm:semi-DI_condition_werner} is similar to the proof of Theorem~\ref{thm:DI_condition_werner}. Theorem~\ref{thm:semi-DI_condition_werner} is derived from the concurrence estimation result for pairs of qubits input, $C_{\rm est}(S)$ in Eq.~\eqref{concurrence_est}, and the condition in Eq.~\eqref{increasing_estimation_condition}. Here we take $\theta_1=\pi/2$ for convenience, when $\theta_2$ and $\theta_3$ satisfy Eq.~\eqref{semi-DI_condition_pure_werner_k} for $k\in[\sqrt{2}/2,1)$, there exists Werner states $\rho_{W}$ such that, the semi-DI concurrence estimation under the settings performs better when taking an $\alpha>1$ CHSH-type Bell expression compared with $\alpha=1$.

We further interpret Theorem~\ref{thm:semi-DI_condition_werner} via a special example. In a semi-DI CHSH Bell test, if measurements in Eq.~\eqref{meas_setting} are set with $\theta_1=\pi/2,\theta_2+\theta_3=0,0<\theta_2<\pi/4$, then any Werner state $\rho_{W}$ with parameter $p$,
\begin{equation}
    p<1-\frac{1}{\sqrt{\sqrt{2}\cos\theta_2\cos(\theta_2-\pi/4)}},
    \label{semi-DI_condition_werner_p}
\end{equation}
promises a better concurrence estimation when taking an $\alpha>1$ CHSH-type Bell expression compared with $\alpha=1$. The right hand side (RHS) of Eq.~\eqref{semi-DI_condition_werner_p} is no large than $1-\frac{1}{\cos(\pi/8)2^{1/4}}$, which implies only Werner state with $p\leq 1-\frac{1}{\cos(\pi/8)2^{1/4}}\doteq 0.0898$ is possible to fit in the condition in Theorem~\ref{thm:semi-DI_condition_werner}. It is worth mentioning that for any $0<\theta_2<\pi/4$, the RHS of Eq.~\eqref{semi-DI_condition_werner_p} is strictly larger than the RHS of  Eq.~\eqref{DI_condition_werner_p}. It allows a wide choice of Werner states to promise a better estimation when taking an $\alpha>1$ CHSH-type Bell expression compared with $\alpha=1$ in the semi-DI experiment. In a semi-DI system with the Werner state and measurements satisfying Eq.~\eqref{semi-DI_condition_werner_p}, the semi-DI concurrence estimation reaches the optimal when the CHSH Bell value is taken at $\alpha$ equals to
\begin{equation}
    \alpha^*_C=\frac{(1-p)^2\cos\theta_2\sin\theta_2}{1-(1-p)^2\cos^2\theta_2}
\end{equation}
and the optimal semi-DI estimation value is
\begin{equation}
    C_{\rm est}|_{\alpha^*_C,{\rm semi-DI}}=\frac{(1-p)\sin\theta_2}{\sqrt{1-(1-p)^2\cos^2\theta_2}}.
\end{equation}
\bibliographystyle{unsrt.bst}

\bibliography{reference.bib}

\begin{thebibliography}{10}

\bibitem{einstein1935can}
A.~Einstein, B.~Podolsky, and N.~Rosen.
\newblock Can quantum-mechanical description of physical reality be considered complete?
\newblock {\em Phys. Rev.}, 47:777--780, May 1935.

\bibitem{bell1964einstein}
J.~S. Bell.
\newblock On the einstein podolsky rosen paradox.
\newblock {\em Phys. Phys. Fiz.}, 1:195--200, Nov 1964.

\bibitem{clauser1969proposed}
John~F. Clauser, Michael~A. Horne, Abner Shimony, and Richard~A. Holt.
\newblock Proposed experiment to test local hidden-variable theories.
\newblock {\em Phys. Rev. Lett.}, 23:880--884, Oct 1969.

\bibitem{schrodinger1935discussion}
Erwin Schr{\"o}dinger.
\newblock Discussion of probability relations between separated systems.
\newblock In {\em Mathematical Proceedings of the Cambridge Philosophical Society}, volume~31, pages 555--563. Cambridge University Press, 1935.

\bibitem{cirelson1980quantum}
B.~S. Cirel'son.
\newblock Quantum generalizations of bell's inequality.
\newblock {\em Letters in Mathematical Physics}, 4(2):93--100, Mar 1980.

\bibitem{guhne2009entanglement}
Otfried Gühne and Géza Tóth.
\newblock Entanglement detection.
\newblock {\em Phys. Rep.}, 474(1):1--75, 2009.

\bibitem{horodecki2009quantum}
Ryszard Horodecki, Pawe\l{} Horodecki, Micha\l{} Horodecki, and Karol Horodecki.
\newblock Quantum entanglement.
\newblock {\em Rev. Mod. Phys.}, 81:865--942, Jun 2009.

\bibitem{curty2004entanglement}
Marcos Curty, Maciej Lewenstein, and Norbert L\"utkenhaus.
\newblock Entanglement as a precondition for secure quantum key distribution.
\newblock {\em Phys. Rev. Lett.}, 92:217903, May 2004.

\bibitem{jozsa2003role}
Richard Jozsa and Noah Linden.
\newblock On the role of entanglement in quantum-computational speed-up.
\newblock {\em Proc. Math. Phys. Eng. Sci.}, 459(2036):2011--2032, 2003.

\bibitem{giovannetti2011advances}
Vittorio Giovannetti, Seth Lloyd, and Lorenzo Maccone.
\newblock Advances in quantum metrology.
\newblock {\em Nat. Photonics}, 5(4):222--229, 2011.

\bibitem{vedral1997quantifying}
V.~Vedral, M.~B. Plenio, M.~A. Rippin, and P.~L. Knight.
\newblock Quantifying entanglement.
\newblock {\em Phys. Rev. Lett.}, 78:2275--2279, Mar 1997.

\bibitem{lo1999unconditional}
Hoi-Kwong Lo and Hoi~Fung Chau.
\newblock Unconditional security of quantum key distribution over arbitrarily long distances.
\newblock {\em Science}, 283(5410):2050--2056, 1999.

\bibitem{shor2000simple}
Peter~W. Shor and John Preskill.
\newblock Simple proof of security of the bb84 quantum key distribution protocol.
\newblock {\em Phys. Rev. Lett.}, 85:441--444, Jul 2000.

\bibitem{bennett1996mixed}
Charles~H. Bennett, David~P. DiVincenzo, John~A. Smolin, and William~K. Wootters.
\newblock Mixed-state entanglement and quantum error correction.
\newblock {\em Phys. Rev. A}, 54:3824--3851, Nov 1996.

\bibitem{raymer1994complex}
M.~G. Raymer, M.~Beck, and D.~McAlister.
\newblock Complex wave-field reconstruction using phase-space tomography.
\newblock {\em Phys. Rev. Lett.}, 72:1137--1140, Feb 1994.

\bibitem{leonhardt1996discrete}
Ulf Leonhardt.
\newblock Discrete wigner function and quantum-state tomography.
\newblock {\em Phys. Rev. A}, 53:2998--3013, May 1996.

\bibitem{leonhardt1997measuring}
Ulf Leonhardt.
\newblock {\em Measuring the quantum state of light}, volume~22.
\newblock Cambridge university press, 1997.

\bibitem{goh2019experimental}
Koon~Tong Goh, Chithrabhanu Perumangatt, Zhi~Xian Lee, Alexander Ling, and Valerio Scarani.
\newblock Experimental comparison of tomography and self-testing in certifying entanglement.
\newblock {\em Phys. Rev. A}, 100:022305, Aug 2019.

\bibitem{xu2014implementation}
Ping Xu, Xiao Yuan, Luo-Kan Chen, He~Lu, Xing-Can Yao, Xiongfeng Ma, Yu-Ao Chen, and Jian-Wei Pan.
\newblock Implementation of a measurement-device-independent entanglement witness.
\newblock {\em Phys. Rev. Lett.}, 112:140506, Apr 2014.

\bibitem{yuan2016reliable}
Xiao Yuan, Quanxin Mei, Shan Zhou, and Xiongfeng Ma.
\newblock Reliable and robust entanglement witness.
\newblock {\em Phys. Rev. A}, 93:042317, Apr 2016.

\bibitem{mayers1998quantum}
Dominic Mayers and Andrew Yao.
\newblock Quantum cryptography with imperfect apparatus.
\newblock In {\em Proceedings of the 39th Annual Symposium on Foundations of Computer Science}, FOCS '98, pages 503--509, Washington, DC, USA, 1998. IEEE Computer Society.

\bibitem{acin2007device}
Antonio Ac\'{i}n, Nicolas Brunner, Nicolas Gisin, Serge Massar, Stefano Pironio, and Valerio Scarani.
\newblock Device-independent security of quantum cryptography against collective attacks.
\newblock {\em Phys. Rev. Lett.}, 98:230501, Jun 2007.

\bibitem{verstraete2002entanglement}
Frank Verstraete and Michael~M. Wolf.
\newblock Entanglement versus bell violations and their behavior under local filtering operations.
\newblock {\em Phys. Rev. Lett.}, 89:170401, Oct 2002.

\bibitem{liang2011semi}
Yeong-Cherng Liang, Tam\'as V\'ertesi, and Nicolas Brunner.
\newblock Semi-device-independent bounds on entanglement.
\newblock {\em Phys. Rev. A}, 83:022108, Feb 2011.

\bibitem{moroder2013device}
Tobias Moroder, Jean-Daniel Bancal, Yeong-Cherng Liang, Martin Hofmann, and Otfried G\"uhne.
\newblock Device-independent entanglement quantification and related applications.
\newblock {\em Phys. Rev. Lett.}, 111:030501, Jul 2013.

\bibitem{toth2015evaluating}
G\'eza T\'oth, Tobias Moroder, and Otfried G\"uhne.
\newblock Evaluating convex roof entanglement measures.
\newblock {\em Phys. Rev. Lett.}, 114:160501, Apr 2015.

\bibitem{arnon2017noise}
Rotem Arnon-Friedman and Henry Yuen.
\newblock Noise-tolerant testing of high entanglement of formation.
\newblock {\em arXiv:1712.09368}, 2017.

\bibitem{chen2018exploring}
Shin-Liang Chen, Costantino Budroni, Yeong-Cherng Liang, and Yueh-Nan Chen.
\newblock Exploring the framework of assemblage moment matrices and its applications in device-independent characterizations.
\newblock {\em Phys. Rev. A}, 98:042127, Oct 2018.

\bibitem{arnon2019device}
Rotem Arnon-Friedman and Jean-Daniel Bancal.
\newblock Device-independent certification of one-shot distillable entanglement.
\newblock {\em New J. Phys.}, 21(3):033010, 2019.

\bibitem{ekert1992quantum}
Artur~K Ekert.
\newblock Quantum cryptography and bell’s theorem.
\newblock In {\em Quantum Measurements in Optics}, pages 413--418. Springer, 1992.

\bibitem{arnon2018practical}
Rotem Arnon-Friedman, Fr{\'e}d{\'e}ric Dupuis, Omar Fawzi, Renato Renner, and Thomas Vidick.
\newblock Practical device-independent quantum cryptography via entropy accumulation.
\newblock {\em Nat. Commun.}, 9(1):1--11, 2018.

\bibitem{zhang2020efficient}
Yanbao Zhang, Honghao Fu, and Emanuel Knill.
\newblock Efficient randomness certification by quantum probability estimation.
\newblock {\em Phys. Rev. Research}, 2:013016, Jan 2020.

\bibitem{zhang2021quantum}
Xingjian Zhang, Pei Zeng, Tian Ye, Hoi-Kwong Lo, and Xiongfeng Ma.
\newblock Quantum complementarity approach to device-independent security.
\newblock {\em arXiv:2111.13855}, 2021.

\bibitem{acin2012randomness}
Antonio Ac\'{\i}n, Serge Massar, and Stefano Pironio.
\newblock Randomness versus nonlocality and entanglement.
\newblock {\em Phys. Rev. Lett.}, 108:100402, Mar 2012.

\bibitem{werner1989quantum}
Reinhard~F. Werner.
\newblock Quantum states with einstein-podolsky-rosen correlations admitting a hidden-variable model.
\newblock {\em Phys. Rev. A}, 40:4277--4281, Oct 1989.

\bibitem{barrett2002nonsequantial}
Jonathan Barrett.
\newblock Nonsequential positive-operator-valued measurements on entangled mixed states do not always violate a bell inequality.
\newblock {\em Phys. Rev. A}, 65:042302, Mar 2002.

\bibitem{peres1999all}
Asher Peres.
\newblock All the bell inequalities.
\newblock {\em Found. Phys.}, 29(4):589--614, 1999.

\bibitem{vertesi2014disproving}
Tam{\'a}s V{\'e}rtesi and Nicolas Brunner.
\newblock Disproving the peres conjecture by showing bell nonlocality from bound entanglement.
\newblock {\em Nat. Commun.}, 5(1):1--5, 2014.

\bibitem{lami2023no}
Ludovico Lami and Bartosz Regula.
\newblock No second law of entanglement manipulation after all.
\newblock {\em Nat. Phys.}, 19(2):184--189, Feb 2023.

\bibitem{wooltorton2022tight}
Lewis Wooltorton, Peter Brown, and Roger Colbeck.
\newblock Tight analytic bound on the trade-off between device-independent randomness and nonlocality.
\newblock {\em Phys. Rev. Lett.}, 129:150403, Oct 2022.

\bibitem{Woodhead2021deviceindependent}
Erik Woodhead, Antonio Ac{\'{i}}n, and Stefano Pironio.
\newblock Device-independent quantum key distribution with asymmetric {CHSH} inequalities.
\newblock {\em {Quantum}}, 5:443, April 2021.

\bibitem{hill1997entanglement}
Sam~A. Hill and William~K. Wootters.
\newblock Entanglement of a pair of quantum bits.
\newblock {\em Phys. Rev. Lett.}, 78:5022--5025, Jun 1997.

\bibitem{rungta2001universal}
Pranaw Rungta, V.~Bu\ifmmode~\check{z}\else \v{z}\fi{}ek, Carlton~M. Caves, M.~Hillery, and G.~J. Milburn.
\newblock Universal state inversion and concurrence in arbitrary dimensions.
\newblock {\em Phys. Rev. A}, 64:042315, Sep 2001.

\bibitem{wilde2017converse}
Mark~M. Wilde, Marco Tomamichel, and Mario Berta.
\newblock Converse bounds for private communication over quantum channels.
\newblock {\em IEEE Trans. Inf. Theory}, 63(3):1792--1817, 2017.

\bibitem{vidal2002computable}
G.~Vidal and R.~F. Werner.
\newblock Computable measure of entanglement.
\newblock {\em Phys. Rev. A}, 65:032314, Feb 2002.

\bibitem{zhang2023quantum}
Xingjian Zhang, Pei Zeng, Tian Ye, Hoi-Kwong Lo, and Xiongfeng Ma.
\newblock Quantum complementarity approach to device-independent security.
\newblock {\em Phys. Rev. Lett.}, 131:140801, Oct 2023.

\bibitem{Navascués2008a}
Miguel Navascués, Stefano Pironio, and Antonio Acín.
\newblock A convergent hierarchy of semidefinite programs characterizing the set of quantum correlations.
\newblock {\em New J. Phys.}, 10(7):073013, jul 2008.

\bibitem{neumark1943representation}
MA~Neumark.
\newblock On a representation of additive operator set functions.
\newblock In {\em CR (Doklady) Acad. Sci. URSS (NS)}, volume~41, pages 359--361, 1943.

\bibitem{Pironio2009device}
Stefano Pironio, Antonio Ac{\'{\i}}n, Nicolas Brunner, Nicolas Gisin, Serge Massar, and Valerio Scarani.
\newblock Device-independent quantum key distribution secure against collective attacks.
\newblock {\em New J. Phys.}, 11(4):045021, apr 2009.

\bibitem{navascues2007bounding}
Miguel Navascu\'es, Stefano Pironio, and Antonio Ac\'{\i}n.
\newblock Bounding the set of quantum correlations.
\newblock {\em Phys. Rev. Lett.}, 98:010401, Jan 2007.

\bibitem{bardyn2009device}
C.-E. Bardyn, T.~C.~H. Liew, S.~Massar, M.~McKague, and V.~Scarani.
\newblock Device-independent state estimation based on bell's inequalities.
\newblock {\em Phys. Rev. A}, 80:062327, Dec 2009.

\bibitem{eberhard1993background}
Philippe~H. Eberhard.
\newblock Background level and counter efficiencies required for a loophole-free einstein-podolsky-rosen experiment.
\newblock {\em Phys. Rev. A}, 47:R747--R750, Feb 1993.

\bibitem{liu2018device}
Yang Liu, Qi~Zhao, Ming-Han Li, Jian-Yu Guan, Yanbao Zhang, Bing Bai, Weijun Zhang, Wen-Zhao Liu, Cheng Wu, Xiao Yuan, et~al.
\newblock Device-independent quantum random-number generation.
\newblock {\em Nature}, 562(7728):548--551, 2018.

\bibitem{li2021experimental}
Ming-Han Li, Xingjian Zhang, Wen-Zhao Liu, Si-Ran Zhao, Bing Bai, Yang Liu, Qi~Zhao, Yuxiang Peng, Jun Zhang, Yanbao Zhang, W.~J. Munro, Xiongfeng Ma, Qiang Zhang, Jingyun Fan, and Jian-Wei Pan.
\newblock Experimental realization of device-independent quantum randomness expansion.
\newblock {\em Phys. Rev. Lett.}, 126:050503, Feb 2021.

\bibitem{arnon2021upper}
Rotem Arnon-Friedman and Felix Leditzky.
\newblock Upper bounds on device-independent quantum key distribution rates and a revised peres conjecture.
\newblock {\em IEEE Trans. Inf. Theory .}, 67(10):6606--6618, 2021.

\bibitem{masanes2006asymptotic}
Llu\'{\i}s Masanes.
\newblock Asymptotic violation of bell inequalities and distillability.
\newblock {\em Phys. Rev. Lett.}, 97:050503, Aug 2006.

\bibitem{designolle2019incompatibility}
S{\'e}bastien Designolle, M{\'a}t{\'e} Farkas, and J{\k{e}}drzej Kaniewski.
\newblock Incompatibility robustness of quantum measurements: a unified framework.
\newblock {\em New J. Phys.}, 21(11):113053, 2019.

\bibitem{buscemi2011entanglement}
Francesco Buscemi and Nilanjana Datta.
\newblock Entanglement cost in practical scenarios.
\newblock {\em Phys. Rev. Lett.}, 106:130503, Mar 2011.

\bibitem{horodecki1995violating}
R.~Horodecki, P.~Horodecki, and M.~Horodecki.
\newblock Violating bell inequality by mixed spin-12 states: necessary and sufficient condition.
\newblock {\em Phys. Lett. A}, 200(5):340--344, 1995.

\bibitem{elkouss2016nearly}
David Elkouss and Stephanie Wehner.
\newblock (nearly) optimal p values for all bell inequalities.
\newblock {\em npj Quantum Inf.}, 2(1):1--8, 2016.

\end{thebibliography}

\end{document}